\documentclass[11pt]{article}

\usepackage[latin1]{inputenc}
\usepackage{amsmath}
\usepackage{amsfonts}
\usepackage{amssymb}
\usepackage{amsthm}

\usepackage{fullpage}

\usepackage{cite}

\usepackage{epsfig}
\usepackage[usenames,dvipsnames]{color}
\usepackage{tabularx}

\usepackage{textcomp}
\usepackage{marvosym}
\usepackage{stmaryrd}
\usepackage{wasysym}
\usepackage{float}

\usepackage{comment}

\newcommand{\ens}[1]{\ensuremath{#1}}
\newcommand{\ith}{\ens{i^{\mbox{\hspace{.2mm}\scriptsize th}}}}
\newcommand{\jth}{\ens{j^{\mbox{\hspace{.2mm}\scriptsize th}}}}
\newcommand{\kth}{\ens{k^{ \mbox{\hspace{.2mm}\scriptsize th}}}}

\newcommand{\cwid}[1]{}

\newcommand{\Ds}{Mergeable Dictionary}
\newcommand{\prob}{Mergeable Dictionary}
\newcommand{\probs}{Mergeable Dictionaries}

\definecolor{litegr}{rgb}{.7,.7,.7}
\definecolor{goodtogo}{rgb}{.7,0,0}
\definecolor{active}{rgb}{0,.75,0}
\definecolor{old}{rgb}{.55,.55,.0}
\definecolor{stylenote}{rgb}{.2,.4,.7}

\newcommand{\kwMs}{Make-Set}
\newcommand{\kwSpl}{Split}
\newcommand{\kwJoinadj}{Join}
\newcommand{\kwUnion}{Merge}
\newcommand{\kwSrc}{Search}
\newcommand{\kwFind}{Find}

\newcommand{\Ms}{\mbox{\textsc{\kwMs}}}

\newcommand{\Spl}{\mbox{\textsc{\kwSpl}}}
\newcommand{\Splx}[2]{\mbox{\textsc{\kwSpl(\ensuremath{#1,#2})}}}
\newcommand{\Joinadj}{\mbox{\textsc{\kwJoinadj}}}

\newcommand{\Union}{\mbox{\textsc{\kwUnion{}}}}
\newcommand{\Unionx}[2]{\mbox{\textsc{\kwUnion(\ensuremath{#1,#2})}}}
\newcommand{\Src}{\mbox{\textsc{\kwSrc}}}
\newcommand{\Srcx}[2]{\mbox{\textsc{\kwSrc(\ensuremath{#1,#2})}}}
\newcommand{\Find}{\mbox{\textsc{\kwFind}}}
\newcommand{\Findx}[1]{\mbox{\textsc{\kwFind(\ensuremath{#1})}}}

\newcommand{\Bslfsrc}{\mbox{\textsc{BSL-FSearch}}}
\newcommand{\Bslfsrcx}[2]{\mbox{\textsc{BSL-FSearch(\ensuremath{#1,#2})}}}
\newcommand{\Bslfspl}{\mbox{\textsc{BSL-FSplit}}}
\newcommand{\Bslfsplx}[1]{\mbox{\textsc{BSL-FSplit(\ensuremath{#1})}}}
\newcommand{\Bslfjoin}{\mbox{\textsc{BSL-FJoin}}}
\newcommand{\Bslfjoinx}[2]{\mbox{\textsc{BSL-FJoin(\ensuremath{#1,#2})}}}
\newcommand{\Bslfrew}{\mbox{\textsc{BSL-FRew}}}
\newcommand{\Bslfrewx}[2]{\mbox{\textsc{BSL-FRew(\ensuremath{#1,#2})}}}

\newcommand{\Bslsrcx}[2]{\mbox{\textsc{BSL-Search(\ensuremath{#1,#2})}}}
\newcommand{\Bslspl}{\mbox{\textsc{BSL-Split}}}
\newcommand{\Bslsplx}[2]{\mbox{\textsc{BSL-Split(\ensuremath{#1,#2})}}}

\newcommand{\Bslrew}{\mbox{\textsc{BSL-Rew}}}
\newcommand{\Bslrewx}[3]{\mbox{\textsc{BSL-Rew(\ensuremath{#1,#2,#3})}}}

\newcommand{\opcount}{4}

\newcommand{\opcountmonetext}{three}
\newcommand{\opcountmtwotext}{two}
\newcommand{\segments}{{\segment}s}
\newcommand{\segment}{segment}

\newcommand{\Segment}{Segment}
\newcommand{\cons}[1]{\ensuremath{c_{#1}}}
\newcommand{\cnsa}{\cons a}
\newcommand{\cnsb}{\cons b}
\newcommand{\cnsc}{\cons c}
\newcommand{\cnsd}{\cons d}
\newcommand{\cnse}{\cons e}
\newcommand{\cnsf}{\cons f}
\newcommand{\cnsg}{\cons g}
\newcommand{\cnsh}{\cons h}

\newcommand{\cnsj}{\cons j}

\newcommand{\invzero}{\ensuremath{\mathbf{(I0)}}}
\newcommand{\invone}{\ensuremath{\mathbf{(I1)}}}
\newcommand{\invtwo}{\ensuremath{\mathbf{(I2)}}}
\newcommand{\isodd}{\ensuremath{\sigma}}
\newcommand{\lasta}{z}
\newcommand{\lastb}{v}
\newcommand{\universe}{\ensuremath{\mathcal U}}
\newcommand{\maxrank}{\ensuremath{R}}
\newcommand{\prevheight}{\ensuremath{h'}}

\newcommand{\nodeheight}[1]{\ensuremath{h(#1)}}
\newcommand{\intheight}[1]{\ensuremath{H(#1)}}
\newcommand{\intmax}[1]{\ensuremath{\max(\intl #1)}}
\newcommand{\intmin}[1]{\ensuremath{\min(\intl #1)}}
\newcommand{\intsubs}[2]{\ensuremath{I(\set{#1},\set{#2})}}
\newcommand{\intnum}[2]{\ensuremath{|I(\set{#1},\set{#2})|}}
\newcommand{\intl}[1]{\ensuremath{#1}}
\newcommand{\beforeop}[1]{\ensuremath{#1'}}
\newcommand{\intweight}[1]{{\ensuremath{W(\intl #1)}}}
\newcommand{\interweight}[3]{{\ensuremath{W_{[#2,#3]}(#1)}}}
\newcommand{\setweight}[1]{{\ensuremath{W(\set #1)}}}
\newcommand{\nodeweight}[1]{{\ensuremath{w(\node #1)}}}
\newcommand{\nodenewweight}[1]{{\ensuremath{w'(\node #1)}}}
\newcommand{\noderank}[1]{\ensuremath{r(\node #1)}}
\newcommand{\nodenewrank}[1]{\ensuremath{r'(\node #1)}}
\newcommand{\node}[1]{\ensuremath{#1}}
\newcommand{\nodeposition}[2]{pos\ensuremath{_{\set #1}(\node #2)}}
\newcommand{\set}[1]{\ensuremath{#1}}
\newcommand{\collection}[1]{\ensuremath{\mathcal #1}}
\newcommand{\numgap}[2]{\ensuremath{g_{\set #1}(#2)}}
\newcommand{\agap}[1]{\ensuremath{a_{#1}}}
\newcommand{\leftofa}[1]{\ensuremath{a_{#1}'}}
\newcommand{\rightofa}[1]{\ensuremath{a_{#1}''}}
\newcommand{\bgap}[1]{\ensuremath{b_{#1}}}
\newcommand{\leftofb}[1]{\ensuremath{b_{#1}'}}
\newcommand{\rightofb}[1]{\ensuremath{b_{#1}''}}
\newcommand{\maxgap}[1]{\ensuremath{#1^{+}}}
\newcommand{\mingap}[1]{\ensuremath{#1^{-}}}
\newcommand{\leftof}[1]{\ensuremath{\gap #1'}}
\newcommand{\rightof}[1]{\ensuremath{\gap #1''}}
\newcommand{\intfun}[1]{\ensuremath{F(\intl #1)}}
\newcommand{\gap}[1]{\ensuremath{#1}}
\newcommand{\potlossa}[1]{\ens{\mathsf{pl}(\agap{#1})}}
\newcommand{\potlossb}[1]{\ens{\mathsf{pl}(\bgap{#1})}}

\newcommand{\maxpotlossa}[1]{\ensuremath{\alpha_{#1}}}
\newcommand{\maxpotlossb}[1]{\ensuremath{\beta_{#1}}}
\newcommand{\prebsltempl}[4]{\ensuremath{#1{[\raisebox{#4}{$#3\mapsfrom\hspace{-0.0mm}#2$}]}}}
\newcommand{\prebsl}[2]{\ensuremath{\prebsltempl{#1}{#2}{\scriptstyle}{0.3mm}}}
\newcommand{\prebslsub}[2]{\ensuremath{\prebsltempl{#1}{#2}{\scriptscriptstyle}{0.2mm}}}
\newcommand{\postbsltempl}[4]{\ensuremath{#1{[\raisebox{#4}{$#3#2\shortrightarrow\hspace{-0.0mm}$}]}}}
\newcommand{\postbsl}[2]{\ensuremath{\postbsltempl{#1}{#2}{\scriptstyle}{0.3mm}}}

\newcommand{\bsl}[1]{\ensuremath{#1}}
\newcommand{\pre}[2]{\text{pred}\ensuremath{_{#1}(#2)}}
\newcommand{\suc}[2]{\text{succ}\ensuremath{_{#1}(#2)}}
\newcommand{\lvlpre}[2]{\ensuremath{L_{#1}(#2)}}
\newcommand{\lvlsuc}[2]{\ensuremath{R_{#1}(#2)}}
\newcommand{\datast}[1]{\ensuremath{D_{#1}}}
\newcommand{\indatast}[1]{\ensuremath{L_{#1}}}

\newcommand{\hide}[1]{}
\newcommand{\locpotfun}[1]{\ensuremath{\Phi_{e}(#1)}}
\newcommand{\potfun}[1]{\ensuremath{\Phi(#1)}}
\newcommand{\plasetclosed}[3]{\ensuremath{S_{[#1,#2]}^#3}}
\newcommand{\plasetopen}[3]{\ensuremath{S_{(#1,#2)}^#3}}
\newcommand{\plasetsmall}[1]{\ensuremath{\plasetopen{0}{a}{#1}}}
\newcommand{\plaseta}[1]{\ensuremath{\plasetclosed{a}{a}{#1}}}

\newcommand{\plasetb}[1]{\ensuremath{\plasetclosed{b}{b}{#1}}}
\newcommand{\plasetfat}[1]{\ensuremath{\plasetopen{b}{2b}{#1}}}
\newcommand{\plasetobese}[1]{\ensuremath{\plasetclosed{2b}{2b+1}{#1}}}
\newcommand{\plats}[1]{\ensuremath{\mathcal P_{#1}}}
\newcommand{\vionum}[1]{\ensuremath{v_{#1}}}
\newcommand{\amcost}[1]{\ensuremath{\hat{c_{#1}}}}
\newcommand{\actcost}[1]{\ensuremath{c_{#1}}}
\newcommand{\op}[1]{#1}
\newcommand{\initbsl}[1]{\ensuremath{\Lambda_{#1}}}
\newcommand{\nodesatlevel}[1]{\ensuremath{N_{h}{(#1)}}}
\newcommand{\ranksatlevel}[1]{\ensuremath{N_{r}{(#1)}}}

\newcounter{count}
\newtheorem{thm}[count]{Theorem}
\newtheorem{lemma}[count]{Lemma}
\newtheorem{defn}[count]{Definition}

\title{\probs}
    
\author{John Iacono\thanks{Research supported by NSF grant CCR-0430849 and by a Alfred P.~Sloan Fellowship.} 
\and
\"Ozg\"ur \"Ozkan\thanks{Research supported by US Department of Education Grant P200A090157.}
}

\date{Department of Computer Science and Engineering\\ Polytechnic Institute of NYU\\2 MetroTech Center\\Brooklyn, NY 11201-3840 USA}

\begin{document}

\maketitle

\begin{abstract} 
A data structure is presented for the \Ds{} abstract data type, which supports the following operations on a collection of disjoint sets of totally ordered data: 
\op{Predecessor-Search}, \op{\kwSpl} and \op{\kwUnion}. While \op{Predecessor-Search} and \op{\kwSpl} work in the normal way, the novel operation is \op{\kwUnion}. While in a typical mergeable dictionary (e.g.~2-4 Trees), the \op{\kwUnion} operation can only be performed on sets that span disjoint intervals in keyspace, the structure here has no such limitation, and permits the merging of arbitrarily interleaved sets. 
Tarjan and Brown present a data structure~\cite{journals/jacm/BrownT79} which can handle arbitrary \op{\kwUnion} operations in $\mathcal O(\log n)$ amortized time per operation if the set of operations is restricted to exclude the \op{\kwSpl} operation. In the presence of \op{\kwSpl} operations, the amortized time complexity of their structure becomes $\Omega(n)$. A data structure which supports both \op{\kwSpl} and \op{\kwUnion} operations in $\mathcal O(\log^2 n)$ amortized time per operation was given by Farach and Thorup~\cite{journals/algorithmica/FarachT98}. 
In contrast, our data structure supports all operations, including \op{\kwSpl} and \op{\kwUnion}, in $O(\log n)$ amortized time, thus showing that interleaved \op{\kwUnion} operations can be supported at no additional cost vis-\`{a}-vis disjoint \op{\kwUnion} operations.
\end{abstract}

\section{Introduction} 

Consider the following operations on a data structure which maintains a dynamic collection \collection S of disjoint sets $\{\set{S_1}, \set{S_2}, \ldots\}$ which partition some totally ordered universal set \universe: 

\begin{itemize} 

\item $S\leftarrow \Findx{x}$: Returns the set $\set S \in \collection S$ that contains $x$. 

\item $p\leftarrow \Srcx{S}{x}$: Returns the largest element in \set S that is at most $x$. 

\item $(\set A,\set B)\leftarrow \Splx{S}{x}$: Splits \set S into two sets $\set A = \{y\in\set S\,|\,y\leq x\}$ and $\set B = \{y\in\set S\,|\,y> x\}$. \set S is removed from \collection S while \set A and \set B are inserted. 

\item $C\leftarrow \Unionx{A}{B}$: Creates $\set C = \set A \cup \set B$. \set C is inserted into \collection S while \set A and \set B are removed. 

\end{itemize} 

We call a data structure that supports these\hide{ \opcount{}} operations a \emph{\Ds{}}. In this paper we present a data structure, which implements these operations in amortized time $\mathcal O(\log n)$, where $n$ is the total number of items in \universe{}. What makes the concept of a \Ds{} interesting is that the \Union{} operation does not require that the two sets being merged occupy disjoint intervals in keyspace. As we discuss in full detail in Section~\ref{sec:previous}, a data structure for merging arbitrarily interleaved sets has appeared independently in the context of Union-Split-Find, Mergeable Trees and string matching in Lempel-Ziv compressed text. In all three cases, a $o(\log^2 n)$ bound on mergeable dictionary operations could not be achieved. We present a data structure that is able to break through this bound though use of a novel weighting scheme applied to an extended version of the Biased Skip List data structure~\cite{journals/algorithmica/BagchiBG05}. Another alternative would be to extend and use Biased Search Trees~\cite{journals/siamcomp/BentST85} but we believe extending Biased Skip Lists will be easier, at least in terms of presentation.

We first present a high-level description of the core data structure of the previous work, and show at a high level the method and motivation we use to improve the runtime. Given this description, we then can discuss in some detail the three aforementioned works. Finally, we present the full details of our result.

\subsection{High-Level Description} 

The basic idea of the structure is simple. As a first attempt we show how to achieve $\mathcal O(\log^2 n)$ time, which we outline here and fully present in Section~\ref{sec:LogSQ}. Store each set using an existing dictionary that supports \Src{}, \Spl{} and \Joinadj{}\footnote{\Joinadj{} merges two sets but requires that the sets span disjoint intervals in keyspace} in $\mathcal O(\log n)$ time (e.g.~2-4 trees). Thus, the only operation that requires a non-wrapper implementation is \Union{}. One first idea would be to implement \Union{} in linear time as in \emph{Merge-Sort}, but this performs poorly, as one would expect. A more intelligent idea is to use a sequence of searches to determine how to partition the two sets into sets of \emph{\segments{}} that span maximal disjoint intervals. Then, use a sequence of \Spl{}s to split each set into the \segments{} and a sequence  \Joinadj{} operations to piece together the \segments{} in sorted order. As the number of \segments{} between two sets being merged could be $\Theta(n)$, the worst-case runtime of such an implementation is $\mathcal O(n \log n)$, even worse than the $\mathcal O(n)$ of a brute-force merge. However, it is impossible to perform many \Union{}s with a high number of \segments{}, and an amortized analysis bears this out; there are only $\mathcal O(\log n)$ amortized \segments{} per \Union{}. Thus, since each \segment{} can be processed in $\mathcal O(\log n)$ time, the total amortized cost per \Union{} operation is $\mathcal O(\log^2 n)$. 

In \cite{klaithesis}, it was shown that there are sequences of operations that have $\Theta(\log n)$ amortized \segments{} per \Union{}. This, combined with the worst-case lower bound 
of $\Omega(\log n)$ for the dictionary operations needed to process each \segment{} seemingly gives a strong argument for a $\Omega(\log^2 n)$ lower bound, which was formally conjectured by Lai. 
It would appear that any effort to circumvent this impediment would require abandoning storing each set in sorted order. 
We show this is not necessary, as a weighting scheme allows us finesse the balance between the cost of processing each \segment{}, and the number of \segments{} to be processed; we, in essence, prevent the worst-case of these two needed events from happening simultaneously. 
Our scheme, combined with an extended version of Biased Skip Lists, allows us 
to speed up the 
processing of each \segment{} to $\mathcal O(1)$ when there are many of them, yet gracefully degrades to the information-theoretically mandated $\Theta(\log n)$ worst-case time when there are only a constant number of \segments{}. The details, however, are numerous. We show in Section~\ref{subsec:BSLExtOps} how to augment Biased Skip Lists to support weighted finger versions of operations that we need. Given this, in Section~\ref{sec:DS} a full description of our structure is presented, and in Section~\ref{sec:Analysis} the runtime analysis is proved. 

\subsection{Relationship to existing work} 
\label{sec:previous} 

The underlying problem addressed here has come up independently three times in the past, in the context of Union-Split-Find, Mergeable Trees and string matching in Lempel-Ziv compressed text. All three of these results, which were initially done independently of each other, bump up against the same $\mathcal O(\log^2n)$ issue with merging, and all have some variant of the $\mathcal O(\log^2n)$ structure outlined above at their core. While the intricacy of the latter two precludes us claiming here to reduce the squared logarithmic terms in their runtimes, we believe that we have overcome the fundamental obstacle towards this improvement.

\subsubsection{Searching in Lempel-Ziv} 

In the paper of Farach and Thorup, an algorithm is presented for string matching in a Lempel-Ziv compressed string \cite{journals/algorithmica/FarachT98}. They show how to search for a string of length $p$ in a compressed string of length $n$ that was compressed by a factor of $f$ in time $\mathcal O(p + n \log^2 f)$. The algorithm is complex, but at its heart needs a data structure that can hold a set \set S of at most $n$ integers in the range $1 \ldots \universe$, and can perform an operation that shifts all values in some interval in \set S by some specified integer amount, so long as all values remain the range $1 \ldots \universe$. They show how to do this in time $\mathcal O(\log \universe \log {N})$ using a variant of the $\mathcal O(\log^2n)$ structure described above. This variant allows a whole set to be shifted, and thus shifting an interval can be done with two splits and two merges in addition to a shift. The potential function they use to bound the number of \segments{} is the same one as in our presentation of the $\mathcal O(\log^2 n)$ structure. A simple extension of our structure can speed up the needed shifts by a 
$\log n$ factor. However, since they do not completely use the aforementioned shifting data structure as a black box (there is an ``unwinding" of some work performed there, among other subtleties), we leave the improvement of the runtime to $\mathcal O(p + n \log f)$ only as a conjectured application of our result. Recently, Pawe\l~\cite{DBLP:conf/esa/Gawrychowski11} achieved this bound using a different approach utilizing grammar-based compression.

\subsubsection{Mergable Trees}

In the paper of Georgiadis, Tarjan, Werneck \cite{conf/soda/GeorgiadisTW06} and the follow-up tech report which adds the authors Kaplan and Shafrir \cite{DBLP:journals/talg/GeorgiadisKSTW11}, the problem of \emph{Mergeable Trees} is studied. This paper is concerned with maintaining a dynamic collection of heaps, subject to operations which constitute the \emph{dynamic tree} ADT: \op{parent}, \op{root}, \op{nca}, \op{insert}, \op{link}, \op{cut}, \op{delete}, and one additional operation, \Union{}, which makes things interesting. In the \Union{} operation, two nodes are specified, and the two root-to-node paths, which are each in sorted order due to the heap property, are merged. The idea of a \emph{Mergeable Tree} ADT originated in an algorithm to compute the structure of a 2-maniford in $R^3$. (The original paper has an algorithm which is shown in \cite{journals/dcg/AgarwalEHW06} to take time $\Omega(\sqrt{n})$ per operation.) 
This ADT is a generalization of our \Ds{}, as if the heaps are restricted to be paths, both structures have identical functionality. They too obtain a $\mathcal O(\log n)$ amortized bound on the number of \segments{}, using the same potential function as in our $\mathcal O(\log^2n)$ presentation, albeit using link-cut trees \cite{journals/jcss/SleatorT83} as the underlying $\mathcal O(\log n)$ structure due to their need to have the operations be performed on a tree paths in a heap rather than a totally ordered set, thus obtaining a $\mathcal O(\log^2 n)$ amortized bound on their operations. We conjecture that using the weighting scheme presented in this paper will allow the development of a Mergeable Tree with $\mathcal O(\log n)$ runtime for all operations. This would require the development of a weighted variant of link-cut trees that support weighted finger searches. In effect, link-cut trees extend $(2,4)$ trees to allow operations on a tree topology, while biased skip lists allow sophisticated weighting operations. We would need a combination of both in order to achieve $\mathcal O(\log n)$ amortized time Mergable Trees. While we conjecture such a combination is possible, the details to be worked out are numerous. 

\subsubsection{Union-Split-Find} 

In the Master's thesis of Lai, the problem of \emph{Union-Split-Find} is studied \cite{klaithesis}. This is proposed as a variant of the classic \emph{Union-Find} data structure of \cite{journals/jacm/Tarjan75}. Our \Ds{} ADT implements \Find{} as \Src{} which is a stronger operation than \Find{} (e.g. \Find{} can be implemented by returning the maxima of a set). 
They propose the $\mathcal O(\log^2 n)$ structure outlined above, and show that its amortized performance is $\Omega(\log^2n)$. They conjecture (correctly) that there is a potential function that gives $\mathcal O(\log^2n)$ runtime, but do not discover it, instead listing several potential functions which do not work.  They show that there is a lower bound of $\Omega(\log n)$ for this problem which follows from dynamic connectivity lower bounds of Mihai P\v{a}tra\c{s}cu \cite{conf/stoc/PatrascuD04}. This lower bound indicates our use of the stronger $\Src{}$ instead of the weaker $\Find{}$ comes at no additional asymptotic amortized cost. They conjecture (incorrectly) that this problem has a lower bound of $\Omega(\log^2n)$ per operation; our $\mathcal O(\log n)$ result refutes this. Our results directly provide an amortized optimal $\mathcal O(\log n)$ time solution to their problem, while the best upper bound they could prove was $\mathcal O(n)$.

\subsection{Future Work} 

This paper opens up several avenues for future work. First, as previously stated we believe that our approach can remove a $\log n$ factor from the runtimes of the Lempel-Ziv searching algorithm and of Mergeable Trees. Both of these results will require integrating our result into each of these different and complicated results. But, since in both cases the same potential function is used as in the $\mathcal O(\log^2 n)$ structure presentation below, we believe that our data structure provides the fundamental breakthrough needed to improve these results.

Secondly, for simplicity we do not consider the dynamic case: our data structure always stores a collection of sets that partitions the same totally ordered set. Extending our result to allow insertion and deletion will require adding additional complexity to our weighting scheme. Finally, we note that the idea that arbitrarily interleaved dictionaries can be merged in $\mathcal O(\log n)$ amortized time is probably a surprising one, which at first glance probably appears to to impossible (which is supported by Lai's inability to get a $o(n)$ solution). While previous results had elements of the \Ds{} ADT, and a $\mathcal O(\log^2 n)$ solution to it, here is the first clear abstraction of it as a pure dictionary problem. 
The fact that the three results above were initially discovered independently of each other also speaks to the ``buried'' nature of the fundamental problem in the previous work. 
We hope that this result finds applications which were not considered because of the seeming impossibility at first glance of a $o(n)$ solution.

\section{A Simple Heuristic for \Union{} and the ${\cal O}(\log^2 n)$ Amortized Bound} 
\label{sec:LogSQ}

As mentioned previously, the main difficulty in designing a data structure for our problem with  $o(n)$ worst-case time complexity lies in being able to perform the \Union{} operation fast. This is confirmed by a lower bound of $\Omega(n)$ on the worst-case time complexity of merging two arbitrary sets \cite{conf/soda/DemaineLM00}. We will describe a heuristic for the \Union{} operation presented in \cite{conf/soda/DemaineLM00} and used in previous work \cite{DBLP:journals/talg/GeorgiadisKSTW11, klaithesis}, and show that the use of this heuristic yields $o(n)$ amortized bounds as a warm up.

Consider the \Unionx{A}{B} operation and a maximal subset in either set \set A or \set B such that all the elements of the other set are less than or greater than each element of the subset. We call this maximal subset a \segment. We can view the \Unionx{A}{B} operation as gluing the appropriate \segments{} of set \set A and \set B. Consider, for instance, the \textit{Merge} algorithm of \textit{Merge-Sort}, which could be used to implement the \Union{} operation. The \textit{Merge} algorithm linearly scans each \segment{} until it locates its maximum element. The \textit{\segment{} merging heuristic} is based on that idea that there are more efficient methods of locating the maximum element of a set than a linear scan.

Next, we make the notion of \segments{} slightly more precise, describe the heuristic, and describe the potential function which yields an upper bound of amortized ${\cal O}(\log^2 n)$ time.

\subsection{The \Segment{} Merging Heuristic} 
\label{subsec:LogSQIntSubApp} 

Define a \emph{\segment} of the \Unionx{A}{B} operation to be a maximal subset \intl S of either set \set A or set \set B such that no element in $(\set A \cup \set B) \setminus \intl S$ lies in the interval $[\intmin{S},\intmax{S}]$. 
  
Each set in the collection is stored as a balanced search tree (i.e.~2-4 tree) with level links. The \hide{\Ms{}, }\Find{}, \Src{}, and \Spl{} operations are implemented\footnote{See \cite{klaithesis} for a detailed description of this implementation.} in a standard way to run in ${\cal O}(\log n)$ worst-case time. The \Unionx{A}{B} operation is performed as follows: 
  
We first locate the minimum and maximum element of each \segment{} of the \Unionx{A}{B} operation using the \Src{} operation and the level links, then extract all these \segments{} using the \Spl{} operation, and finally we merge all the \segments{} in the obvious way using a standard \Joinadj{} operation. Therefore since each operation takes ${\cal O}(\log n)$ worst-case time, the total running time is ${\cal O}(T\cdot\log n)$ where $T$ is the number of \segments{}. 
  
We now analyze all the operations using the potential method \cite{amortizedcomplexity}, with respect to two parameters: $n$,\hide{the total number of \Ms{} operations} the size of the universe, and $m$, the total number of all operations. 
  
Let \datast i represent the data structure after operation $i$, where \datast 0 is the initial data structure. Operation $i$ has a cost of \actcost i and transforms $\datast{i-1}$ into $\datast i$. We have a potential function $\Phi : \{\datast i\} \rightarrow\mathbb{R}$ such that $\potfun{\datast 0} = 0$ and $\potfun{\datast i}\geq 0$ for all $i$. The amortized cost of operation $i$, $\amcost i$, with respect to $\Phi$ is defined as $\amcost i = \actcost i + \potfun{\datast i} - \potfun{\datast{i-1}}$. The total amortized cost of $m$ operations will be 
\[ 
\sum_{i=1}^m \amcost i = \sum_{i=1}^m (\actcost i + \potfun{\datast i} - \potfun{\datast{i-1}})= \sum_{i=1}^m \actcost i + \potfun{\datast n} - \potfun{\datast 0} 
\geq \sum_{i=1}^m \actcost i 
\] 
since $\potfun{\datast n}\geq 0$ and $\potfun{\datast 0} = 0$. Thus, the amortized cost will give us an upper bound on the worst-case cost.

Next, we describe a potential function which yields an amortized bound of ${\cal O}(\log^2 n)$ on the running time. This potential function was essentially used in \cite{DBLP:journals/talg/GeorgiadisKSTW11, journals/algorithmica/FarachT98} which are the only instances where a $o(n)$ solution has been presented.

\subsection{The Potential Function} 
\label{subsec:LogSQPotFun} 
We need to define some terminology before describing the potential function. Let \nodeposition{S}{x} be the position of $x$ in set $S$, or more formally $\nodeposition{S}{x} = |\{y\in \set S\,|\,y\leq x\}|$. Then $\numgap{S}{k}$, the size of the \kth\ gap of set \set S, is the difference of positions between the element of position $k$ and $k+1$ of set \set S in universe $U$. In other words, $\numgap{S}{k} = \nodeposition{U}{x} - \nodeposition{U}{y}$ where $\nodeposition{S}{x} = k$ and  $\nodeposition{S}{y} = k+1$. For the boundary cases, let $\numgap{S}{0} = \numgap{S}{|S|} = 1$. 
  
Recall that $D_i$ is the data structure containing our dynamic collection of disjoint sets, ${\cal S}^{(i)} = \{S_1^{(i)}, S_2^{(i)},\ldots\}$ after the \ith\ operation. Finally, let $\varphi(S) = \sum_{j=1}^{|S|-1} \log \numgap{S}{j}$. Then we define the potential after the \ith\ operation as follows: 
\[ 
\potfun{\datast i} = \cnsa\cdot\sum_{\set S \in \collection{S^{(i)}}} \varphi\left(S\right)\log n 
\] 
where \cnsa{} is a positive constant to be determined later. 
  
Note that since the collection of sets consists of the $n$ singleton sets, the data structure initially has 0 potential ($\potfun{\datast 0}=0$). Furthermore, because any gap has size at least 1, the data structure always has non-negative potential ($\potfun{\datast i} \geq 0,\ \forall i\geq 0$).

\subsection{The Amortized ${\cal O}(\log^2 n)$ Bound} 
\label{subsec:LogSQBound} 

The \hide{\Ms{}, }\Find{}, \Src{}, and \Spl{} operations have worst-case ${\cal O}(\log n)$ running times. The first \opcountmtwotext{} of these operations do not change the structure and therefore do not affect the potential. Observe that the \Spl{} operation can only decrease the potential. Thus, the amortized cost of all \opcountmonetext{} operations is ${\cal O}(\log n)$. 

Now, suppose the \ith\ operation is \Unionx{A}{B} where \set $A$ and \set $B$ are sets in \datast{i-1}. Assume w.l.o.g.~that the minimum element in $\set A\cup \set B$ is an element of \set $A$. Let $\intsubs{A}{B} = \{\intl A_1, \intl B_1, \intl A_2, \intl B_2,\ldots \}$ be the set of \segments{} of operation \Unionx{A}{B}, where $\intmax{A_i} < \intmin{A_j}$ and $\intmax{B_i} < \intmin{B_j}$ for $i<j$, and $\intmax{A_i}<\intmin{B_i}<\intmax{B_i}<\intmin{A_{i+1}}$ for all $i$. 
As previously noted, the worst-case cost of the \Union{} operation is ${\cal O}(\intnum{A}{B}\cdot\log n)$. Let $\agap i$ be the size of the gap between the maximum element of $\intl A_i$ and the minimum element of $\intl A_{i+1}$, or more formally let $\agap i = \numgap{\intl A_i \cup \intl A_{i+1}}{|\intl A_i|}$. Define $\bgap i$ similarly. Now, let

\begin{align*} 
\leftofa{i} &= \numgap{\intl A_i \cup \intl B_{i}}{|\intl A_i|}\\ 
\rightofa{i} &= \numgap{\intl B_i \cup \intl A_{i+1}}{|\intl B_i|} 
\end{align*} 
and 
\begin{align*} 
\leftofb{i} &= \numgap{\intl B_i \cup \intl A_{i+1}}{|\intl B_i|}\\ 
\rightofb{i} &= \numgap{\intl A_{i+1} \cup \intl B_{i+1}}{|\intl A_{i+1}|} 
\end{align*}

\begin{figure}
\centering 
\includegraphics[scale=.6]{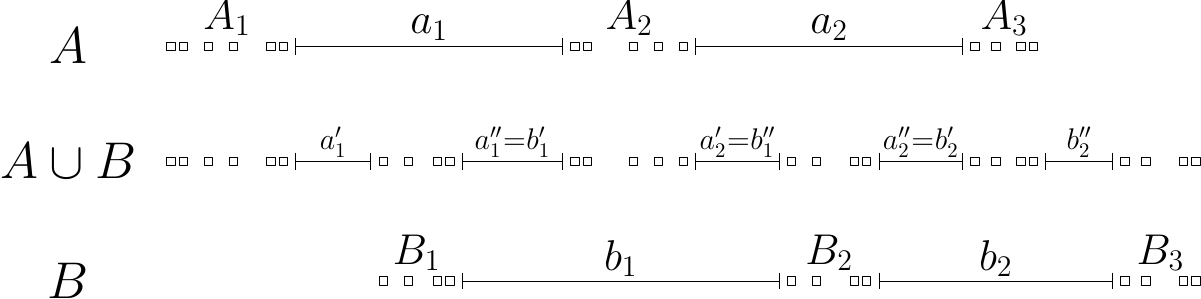} 
\caption{Gaps \agap i, \bgap i, \leftofa i, \rightofb i, \leftofb i, and \rightofb i are defined with respect to the \Unionx{A}{B} operation.} 
\label{fig:lgsqgaps} 
\end{figure} 
 
Note that $\rightofa{ i} = \leftofb{i}$ and $\leftofa{i} = \rightofb{{i-1}}$ (see Figure \ref{fig:lgsqgaps}). During the analysis we will take into account whether \intnum{A}{B} is odd or even. Let $\isodd = \intnum{A}{B} \bmod{2}$, and $R = \lfloor (\intnum{A}{B}-2)/2 \rfloor$. We are now ready to bound the amortized cost of the \Union{} operation. We have

\[\amcost i = \actcost i + \potfun{\datast i} - \potfun{\datast{i-1}}\] 

\begin{align*} 
\potfun{\datast{i-1}} &= \sum_{S\in \collection S \setminus\{\set A, \set B\}} \varphi(S)\cnsa\log n\\ 
& \quad\, +\left(\sum_i \varphi(\intl A_i) +\sum_i \varphi(\intl B_i)\right)\cnsa\log n  \\ 
& \quad\, +\left(\sum_{i=1}^R (\log \agap i + \log \bgap i) + \isodd\log \agap{R+1}\right)\cnsa\log n 
\end{align*} 
and 
\begin{align*} 
\potfun{\datast i} &= \sum_{S\in \collection S \setminus\{\set A, \set B\}} \varphi(S)\cnsa\log n\\ 
& \quad\, +\left(\sum_i \varphi(\intl A_i) +\sum_i \varphi(\intl B_i) + \log \leftofa{1}\right)\cnsa\log n  \\ 
& \quad\, +\frac{1}{2}\left(\sum_{i=1}^R \log \rightofa{ i} + \log \leftofb{i} + \log \rightofb{ i} + \log \leftofa{i+1}\right)\cnsa\log n\\ 
& \quad\, + \frac{1}{2}\left(\log \rightofa{{R+1}} + \log \leftofb{R+1}\right)\isodd \cnsa\log n. 
\end{align*} 
This gives us 
{\allowdisplaybreaks 
\begin{align*} 
\potfun{\datast i} - \potfun{\datast{i-1}} &= \left(\frac{1}{2}\left(\sum_{i=1}^R  \log \rightofa i \leftofb i \rightofb i \leftofa{i+1} - 2\log \agap i \bgap i\right)+\log \leftofa 1 \right)\cnsa\log n\\ 
& \quad\, +\left(\frac{1}{2}\left(\log \rightofa{R+1} + \log \leftofb{R+1}\right) -\log \agap{R+1} \right)\isodd \cnsa\log n  \\ 
&=\left(\frac{1}{2}\left(\sum_{i=1}^R  \log \rightofa i \leftofb i \rightofb i \leftofa i - 2\log \agap i \bgap i\right)+\frac{1}{2}\left(\log \leftofa 1 + \log \leftofa{R+1}\right)\right)\cnsa\log n\\ 
& \quad\, +\left(\frac{1}{2}\left(\log \rightofa{R+1} + \log \leftofb{R+1}\right) -\log \agap{R+1} \right)\isodd \cnsa\log n 
\end{align*} 
We have $\rightofa i, \leftofa i < \agap i \leq n$ and similarly $\rightofb i, \leftofb i < \bgap i \leq n$. Also note that since $\leftofa i + \rightofa i \leq \agap i$, we have $\log \leftofa i + \log \rightofa i \leq  \log \leftofa i + \log (\agap i - \leftofa i)\leq \log \agap i/2 + \log \agap i/2.$ Similarly, $\log \leftofb i + \log \rightofb i \leq \log \bgap i/2 + \log \bgap i/2.$ 
\begin{align*} 
\potfun{\datast i} - \potfun{\datast{i-1}}&\leq \frac{1}{2}\left(\sum_{i=1}^R  \log \rightofa i \leftofb i \leftofb i' \leftofa i - 2\log \agap i \bgap i\right)\cnsa\log n + {\cal O}(\log^2 n)\\ 
&= \frac{1}{2}\left(\sum_{i=1}^R \log \frac{ \leftofa i\cdot \rightofa i\cdot \leftofb i \cdot \rightofb i}{ \agap i \cdot \agap i\cdot \bgap i\cdot \bgap i}\right)\cnsa\log n + {\cal O}(\log^2 n)\\ 
&\leq \frac{1}{2}\left(\sum_{i=1}^R \log \frac{ \agap i/2\cdot \agap i/2\cdot \bgap i/2 \cdot \bgap i/2}{ \agap i\cdot \agap i\cdot \bgap i\cdot \bgap i}\right)\cnsa\log n + {\cal O}(\log^2 n)\\ 
&= \frac{1}{2}\left(\sum_{i=1}^R \log \frac{ 1}{ 16}\right)\cnsa\log n + {\cal O}(\log^2 n)\\ 
&= -\cnsb\cdot I(A,B)\cnsa\log n +{\cal O}(\log^2 n).
\end{align*}}

Recall that the worst-case cost of the \Union{} operation, $\actcost i$, is ${\cal O}(\intnum{A}{B}\log n)$. Let $\cnsc$ be a constant such that $\actcost i \leq \cnsc\cdot\intnum{A}{B}\log n$. Then the above bound yields 
\begin{align*} 
\amcost i &= \actcost i + \potfun{\datast i} - \potfun{\datast{i-1}}\\ 
&\leq \actcost i - \cnsb\cdot \intsubs{A}{B}\cnsa\log n + {\cal O}(\log^2 n)\\ 
&\leq \cnsc\cdot \intsubs{A}{B}\log n - \cnsb\cdot \intsubs{A}{B}\cnsa\log n + {\cal O}(\log^2 n)\\ 
&= {\cal O}(\log^2 n) \qquad\qquad\qquad \mbox{(set }\cnsa=\cnsc/\cnsb). 
\end{align*}

Thus, the amortized cost of the \Union{} operation is ${\cal O}(\log^2 n)$. Combined with the arguments before, this gives us the following theorem.

\begin{thm} 
\label{thm:logsq} 
The \prob{} problem can be solved such that a sequence of $m$ \hide{\Ms{}, }\Find{}, \Src{}, \Union{}, and \Spl{} operations\hide{, $n$ of which are \Ms{} operations, } can be executed in ${\cal O}(m\log^2 n)$  worst-case time. 
\end{thm}

In order to obtain a data structure with an amortized running time of $o(\log^2 n)$ per operation, we certainly need a new potential function. 
To see this, observe the case when we have two singleton sets $\set A$ and $\set B$ with elements $\node x\in \set A$ and $\node y\in \set B$ such that $\log(\nodeposition{U}{x} - \nodeposition{U}{y}) = \Omega(\log n)$. If we use the potential function defined in Section~\ref{subsec:LogSQPotFun}, the potential increase alone as a result of a \Unionx{A}{B} operation is $\Omega(\log^2 n)$. We want to eliminate the extra $\log n$ factor in the potential function but this implies we need to be able to join \segments{} in amortized ${\cal O}(1)$ time. We ultimately want a data structure such that each operation except \Union{} can be performed in worst-case ${\cal O}(\log n)$ time, and \Union{} can be performed in worst-case 
\[ 
\mathcal O\left(\log n + \sum_i \intfun{A_i} + \sum_j \intfun{B_j} \right) 
\] 
time where $\intl A_i$ and $\intl B_j$ are \segments{} involved in the operation, \intfun{A_i} and \intfun{B_j} denote the time it takes to process \segments{} $\intl A_i$  and $\intl B_j$ respectively, and $\sum_i \intfun{A_i}+ \sum_j \intfun{B_j}$ is no more than the decrease in potential.

In the next section, we describe biased skip lists, the underlying data structure we will be using in our data structure.

\section{Biased Skip Lists with Extended Operations} 
\label{sec:BSL} 
Biased skip lists are a variant of skip lists \cite{journals/cacm/Pugh90} which we assume the reader is familiar with. Biased skip lists as described in \cite{journals/algorithmica/BagchiBG05} are missing some operations which will be vital in the implementation of our data structure. Therefore, in order to be able to design a highly tuned \Union{} operation, we will extend biased skip lists. 

First, we describe essential biased skip list details. The reader is referred to \cite{journals/algorithmica/BagchiBG05} for further details on biased skip lists.

\subsection{Biased Skip Lists} 
\label{subsec:BSLBSL}

We will first cover basic definitions followed by the three key invariants of the structure.

\paragraph{Definitions} 
A biased skip list (BSL) $\bsl S$ stores an ordered set $\set X$ where each element $x\in X$ corresponds to a node\footnote{We will use the terms ``element", ``node", ``key", and ``item" interchangeably; the context clarifies any ambiguity.} $\node x\in \bsl S$ with \textit{weight} $\nodeweight x$, which is user-defined, and integral \textit{height} $\nodeheight x$, which is initially computed from the weight of the node. For our purposes, we will assume that the weights are bounded from below by 1 and bounded from above by a polynomial in $n$. 

Each node $\node x \in \bsl S$ is represented by an array of length $\nodeheight x +1$ called the \textit{tower} of node \node x. The \textit{level-j predecessor}, \lvlpre{j}{x}, of $x$ is the largest node $\node k$ in \bsl S such that $\node k<x$ and $\nodeheight k \geq \node j$. The \textit{level-j successor}, \lvlsuc{j}{x}, is defined symmetrically. The \jth\ element of the tower of node \node x, contains pointers to the \jth\ elements of towers of node \lvlpre{j}{x} and node \lvlsuc{j}{x} with the exception of towers of adjacent nodes where pointers between any pair of adjacent nodes $x$ and $y$ on level $\min(\nodeheight x,\nodeheight{y})-1$ are nil and the pointers below this level are undefined. Node levels progress from top to bottom. Two distinct elements $\node x$ and $\node y$ are called \textit{consecutive} if and only if they linked together in $\set S$; or equivalently if and only if for all $\node x<\node z<\node y$, $\nodeheight z< \min(\nodeheight x,\nodeheight y)$. A \textit{plateau} is a maximal set of consecutive nodes of the same height. The \textit{rank} of a node \node x is defined as $\noderank x = \lfloor\log_a \nodeweight x \rfloor$ where $a$ is a constant to be specified later. For our purposes, we will set $a=2$.

Additionally, let $\pre{X}{x}$ be the predecessor of $x$ in set \set X, and let $\suc{X}{x}$ be the successor of $x$ in set \set $X$. Let $\intheight X = \max_{x\in X} \nodeheight x$. Let $\prebsl{S}{j} = \{\node x\in \bsl S\,|\,x\leq j\}$ and $\postbsl{S}{j} = \{\node x\in \bsl S\,|\,x> j\}$. Let $\intweight S= \sum_{\node x\in \set S} \nodeweight x$. Also let $\interweight{S}{i}{j}= \sum_{\node x \in \set S; i\leq \node x \leq j} \nodeweight x$.

For convenience, we imagine sentinel nodes $-\infty$ and $+\infty$ of height \intheight S at the beginning and end of biased skip list \bsl S. These sentinels are not actually stored or maintained. 

 The \textit{left profile} of \node x in a biased skip list \bsl S is defined as $\{\lvlpre{j}{x}\,|\,\nodeheight{\lvlpre{j}{x}} = j\}$. Similarly the \textit{right profile} of \node x in a biased skip list \bsl S is defined as $\{\lvlsuc{j}{x}\,|\,\nodeheight{\lvlsuc{j}{x}} = j\}$. The \textit{profile} of node $x$ in a biased skip list \bsl S is the union of its left profile and right profile. 

 The \textit{left cover} of a biased skip list \bsl S is defined as $\{\intmin{S}\}\cup \{\lvlsuc{j}{\intmin{S}}\,|\,\nodeheight{\lvlsuc{j}{\intmin{S}}} = j,\ j>\nodeheight{\intmin{S}} \}$. 
Similarly, the \textit{right cover} of a biased skip list \bsl S is defined as $\{\intmax{S}\}\cup \{\lvlpre{j}{\intmax{S}}\,|\,\nodeheight{\lvlpre{j}{\intmax{S}}} = j,\ j>\nodeheight{\intmax{S}} \}$. 
The \textit{cover} of biased skip list \bsl S is the union of its left cover and right cover.

\paragraph{Invariants} 
The three invariants of biased skip lists are listed below. Note that $a$ and $b$ can be suitable constants satisfying the definition of $(a,b)$-biased skip lists. For our purposes it is sufficient to set $a=2, b=6$. 
\begin{defn} 
For any $a$ and $b$ such that $1<a\leq \lfloor\frac{b}{3}\rfloor$, an $(a,b)$-biased skip list is a biased skip list with the following properties: 
\begin{description} 
\item {\invzero} Each item $x$ has height $\nodeheight x\geq \noderank x$. 
\item {\invone} There are never more than $b$ consecutive items of any height. 
\item {\invtwo} For each node $x$ and for all $i$ such that $\noderank x < i \leq \nodeheight x$, there are at least $a$ nodes of height $i-1$ between $x$ and any consecutive node of height at least $i$. 
\end{description} 
\end{defn} 
\noindent In the remainder of the paper, we will refer to $(2,6)$-biased skip lists simply as biased skip lists. 
\subsection{Operations} 
\label{subsec:BSLOps} 

We now describe the original biased skip list operations we will be using in our data structure. 

\begin{itemize} 

 \item $p \leftarrow \Bslsrcx{S}{i}$: Performs a standard search in biased skip list $\bsl S$ using search key $i$. This operation runs in worst-case ${\cal O}(\log n)$ time. 
 \item $p\leftarrow \Bslfsrcx{i}{j}$: Starting from a given finger\footnote{This operation is originally described with three arguments in \cite{journals/algorithmica/BagchiBG05} as $FingerSearch(X,i,j)$. However, note that $X$ is a redundant argument here as we already have a pointer to an element of $X$, namely, to $i$.} to a node $\node i$ in some biased skip list $\bsl S$ perform a predecessor search in  $\bsl S$ using $j$ as the search key. This operation runs in 
\[ 
{\cal O}\left(1 + \log\frac{\interweight{S}{i}{\suc{S}{j}}}{\min(\nodeweight{i},\nodeweight{\pre{S}{j}}, \nodeweight{\suc{S}{j}})}\right) 
\] worst case time. 
 \item $(A,\bsl B)\leftarrow \Bslsplx{S}{i}$: Splits the biased skip list $\bsl S$ at $i$ into two biased skip lists \bsl A and \bsl B storing sets $\{x\in S\,|\, x\leq i\}$ and $\{x\in S\,|\,x> i\}$ respectively, and returns an ordered pair of handles to $\bsl A$ and $\bsl B$. This operation runs in worst-case ${\cal O}(\log n)$ time. 

\item \Bslrewx{S}{i}{w}: Changes the weight of node $\node i\in \bsl S$ to $w$. This operation runs in worst-case ${\cal O}(\log n)$ time.

\end{itemize} 

\subsection{Extended Operations} 
\label{subsec:BSLExtOps} 
Biased skip lists support finger searches, however we need to extend biased skip lists to also support finger split, finger join, and finger reweight operations.

\paragraph{Finger Split} Given a pointer to node $\node f\in\bsl S$, \Bslfsplx{f} splits biased skip list $\bsl S$ into two biased skip lists, \bsl A and \bsl B, storing sets $\{\node x\in \bsl S\,|\, \node x\leq \node f\}$ and $\{\node x\in \bsl S\,|\, \node x> \node f\}$ respectively, and returns an ordered pair of handles to $\bsl A$ and $\bsl B$.

\vspace{5mm} 
\noindent $\underline{(A,\bsl B)\leftarrow\Bslfsplx{f}}$: 
\begin{enumerate} 
\item Disconnect the pointers between the each node in the right profile of node \node f and the left profile of node \node $f'=\suc{S}{f}$, effectively splitting \bsl S and forming $\bsl A$ and $\bsl B$. More precisely, disconnect pointers between the \jth\ level of node \lvlsuc{j}{f} and \lvlpre{j}{f'} where $j\geq\min(\nodeheight{f},\nodeheight{f'})$. Pointers below this level are already null. 

\item Restore \invtwo{} in $\bsl A$.
We will process the nodes of the right cover of \bsl A after Step~1 in the order of increasing height. Denote the current node being processed by \node u. Let \prevheight{} be the height of the node which was most recently processed. 
\begin{enumerate} 
\item If $\node u = \node f$, then set $\prevheight = \nodeheight u$ and demote the height of \node u to \noderank u. 

\item  If \invtwo{} is not violated at \node u, then stop if $\nodeheight u>\min(\intheight{\prebsl{S}{f}},\intheight{\postbsl{S}{f}})$, and set $\prevheight = \nodeheight u$ and iterate with the next node, \lvlpre{\nodeheight{u}+1}{A}, otherwise\footnote{Note that \lvlpre{\nodeheight{u}+1}{A} can be computed in $\mathcal O(b) = \mathcal O(1)$ time due to \invone{}; and $\min(\intheight{\prebsl{S}{f}},\intheight{\postbsl{S}{f}})$ can be computed during Step~1.}. 

\item If \invtwo{} is violated at node \node u, then demote the height of $u$ to $\max(\noderank u, \prevheight)$. Set $\prevheight$ to the height of \node u before the demotion. 

\begin{itemize} 
\item If a demotion causes an \invone{} violation at the new height of \node u by creating a plateau of $b'>b$ nodes, then promote the height of the median of these nodes by 1, and iterate at the level above to check for percolating \invone{} violations. 
\item Once all the \invone{} violations are fixed, iterate with the next node, \lvlpre{\prevheight +1}{A}, to fix the next potential \invtwo{} violation. 
\end{itemize} 

\end{enumerate} 

Restore \invtwo{} in $\bsl B$ essentially symmetrically. 
\end{enumerate} 

\paragraph{Correctness} All the pointers in \bsl S connecting any node in \bsl A and any node in \bsl B are precisely those described in Step 1. Therefore Step 1 splits the nodes of $\bsl A$ and $\bsl B$ correctly. 

We need to ensure that all the invariants are preserved. When we perform demotions in Step 2 we make sure that we do not demote the height of any node lower than its rank. Therefore \invzero{} is preserved. Note that removing predecessors or successors cannot cause an \invone{} violation. 

Observe that in $\bsl A$, \invtwo{} can only be violated at the nodes in the right cover of \bsl A after Step~1, once per level. A symmetric argument holds for \invtwo{} violations in \bsl B. 
When we demote a node \node u this fixes the \invtwo{} violation at that node because the height of the node is either demoted to \noderank u which by definition implies the violation is fixed, or to \prevheight{}. 
If the height of node \node u is demoted to \prevheight{}, this implies that there was another node \node{u'} with height \prevheight{}, consecutive to \node u, before the \Bslfspl{} operation. 
The only way we change the height of any node between node \node u and node \node{u'} during \Bslfsplx{f} is an \invone{} promotion, which, cannot cause an \invtwo{} violation. Note that there cannot be an \invtwo{} violation at these nodes since they are not in the right cover of \bsl A after Step~1. We assume there were no \invtwo{} violations before the \Bslfspl{} operation. Then it holds that there can be no \invtwo{} violations at node \node u at any level less than or equal to \prevheight{}. Therefore, demoting \node u to level \prevheight{} fixes any \invtwo{} violations at this node. 

Any of the demotions may cause an \invone{} violation which are also fixed by Step~2. Because we promote the median, this cannot cause an \invtwo{} violation since $\lfloor b'/2\rfloor \geq \lfloor b'/3\rfloor \geq a$. Also note that the \invone{} promotions caused by the demotion of a node never percolate higher than the height of the node before the demotion. Thus, nodes on the right cover of \bsl A after Step~1 that have not yet been processed during Step~2 are not removed from the right cover due to an \invone{} promotion. 

Observe that if a node \node u in the right cover of \bsl A after Step~1 has height greater than $\min(\intheight{\prebsl{S}{f}},$ $\intheight{\postbsl{S}{f}})$ and there is no \invtwo{} violation at this node, then no other node of greater height in \bsl A can have an \invtwo{} violation. Symmetric arguments again apply for \bsl B. Thus, by induction, iterating Step 2 fixes all \invtwo{} violations.

\paragraph{Finger Join} 
Given pointers to $\node \ell$ and $\node r$, the maximum and minimum nodes of two distinct biased skip lists $\bsl A$ and $\bsl B$ respectively, \Bslfjoinx{\ell}{r} returns a new biased skip list $\bsl S$ containing all elements of $\bsl A$ and $\bsl B$ assuming $\node \ell < \node r$. $\bsl A$ and $\bsl B$ are destroyed in the process.

\vspace{5mm} 
\noindent $\underline{\bsl S \leftarrow \Bslfjoinx{\ell}{r}}$:
\begin{enumerate} 
\item Connect pointers between each node in the right cover of \bsl A and the left cover of \bsl B, effectively joining \bsl A and \bsl B, and forming \bsl S. 
More precisely, create pointers between the \jth\ level of node \lvlsuc{j}{r'} and \lvlpre{j}{\ell'} for all $j\geq\min(\nodeheight{\ell},\nodeheight{r})$, where $\ell'$ is the $+\infty$ sentinel node of \bsl A, and $r'$ is the $-\infty$ sentinel node of \bsl B. Pointers below this level need to be null. 

\item Restore \invone{} in \bsl S. For $j = \max(\nodeheight{\ell},\nodeheight{r})$ up through $\min(\intheight A,\intheight B)$, if there is an \invone{} violation at level $j$ caused by a plateau of $b'>b$ nodes, then promote the height of the median of these nodes by 1, and iterate at the next level. For $j = \min(\intheight{A},\intheight{B})+1$ up through $\max(\intheight A,\intheight B)$, if there is no \invone{} violation at level $j$, then stop\footnote{$\min(\intheight{A},\intheight{B})$ can be computed during Step~1.}. Otherwise promote the median node of the plateau of $b'>b$ nodes causing the violation to restore \invone{} and iterate at the next level. 
\end{enumerate} 

\paragraph{Correctness} 
Joining \bsl A and \bsl B only affects the right cover of \bsl A and the left cover of \bsl B. Therefore, Step 1 connects the two biased skip lists \bsl A and \bsl B and forms \bsl S correctly. Joining two biased skip lists cannot create any \invtwo{} violations assuming there were no such violations before the operation; and fixing \invone{} violations cannot create \invtwo{} violations since $\lfloor b'/2\rfloor \geq\lfloor b'/3\rfloor \geq a$. Therefore, \invtwo{} is preserved. Note that after Step 1, any level in the range $[\max(\nodeheight{\ell}, \nodeheight{r}),\max(\intheight{A},\intheight{B})]$ could have at most $2b+1$ ($b$ from \bsl A, $b$ from \bsl B, and $1$ due to a promotion from the level below) consecutive nodes of the same height; which is fixed in Step 2. By induction, iterating Step 2 fixes all \invone{} violations.

\paragraph{Finger Reweight} 
Given a pointer to a node $\node f \in \bsl S$, changes its weight to $w$ while preserving invariants \invzero{}, \invone{}, and \invtwo{} of the biased skip list containing $\node f$.

\vspace{5mm} 
\noindent$\underline{\Bslfrewx{f}{w}}$: 
\begin{enumerate} 
\item Let \nodenewrank f be the new rank of \node f. If $\nodenewrank f = \noderank f$, then stop. 
\item If $\nodenewrank f > \noderank f$ and $\nodeheight f \geq \nodenewrank f$, then stop. 
\item If $\nodenewrank f > \noderank f$ and $\nodeheight f < \nodenewrank f$, then promote the height of $f$ to \nodenewrank f. Restore \invtwo{} as in Step 2 of \Bslfspl{} but start from the first node in the left profile of \node f that has height greater than \nodeheight f; and symmetrically from the first node in the right profile of \node f that has height greater than \nodeheight f. Then, restore \invone{} as in Step 2 of \Bslfjoin{} but start from level \nodenewrank f. 
\item If $\nodenewrank f < \noderank f$, then demote the height of \node f to \nodenewrank f. Restore \invone{} as in Step 2 of \Bslfjoin{} but starting from \nodenewrank f. Then, restore \invtwo{} as in Step 2 of \Bslfspl{}, but starting from the first node in the left profile of \node f that has height greater than \nodeheight f; and symmetrically from the first node in the right profile of \node f that has height greater than \nodeheight f. 
\end{enumerate} 

\paragraph{Correctness} If the rank of node \node f does not change, there are no structural changes to the biased skip list and therefore Step 1 is correct. 

If the rank of \node f increases but it is still less than its height, then \invzero{} is preserved. By not changing the height of the node we ensure that \invone{} and \invtwo{} are preserved as well. Therefore, Step 2 is also correct.

If the rank of a node \node f becomes greater than its height, then \invzero{} is violated and we promote \node f to its new rank to fix the violation. Observe that this promotion can cause \invtwo{} to be violated at the nodes in the left profile of \node f that have height greater than the old height of \node f; and symmetrically at the nodes in the right profile of \node f that have height greater than the old height of \node f, once per level. Step 3, by the correctness of Step 2 of \Bslfspl{}, fixes all \invtwo{} violations. The promotion can also cause \invone{} to be violated at level \nodenewrank f. Step 3, by the correctness of Step 2 of \Bslfjoin{}, fixes all \invone{} violations. Therefore, Step 3 is correct. 

If the rank of a node \node f decreases, we demote \node f to its new rank so \invtwo{} cannot be violated at this node. Observe that this demotion can cause \invone{} to be violated at level \nodenewrank f. Step 4, by the correctness of Step 2 of \Bslfjoin{}, fixes all \invone{} violations. The demotion can also cause \invtwo{} to be violated at the nodes in the left profile of \node f that have height greater than the old height of \node f; and symmetrically at the nodes in the right profile of \node f that have height greater than the old height of \node f, once per level. Step 4, by the correctness of Step 2 of \Bslfspl{}, fixes all \invtwo{} violations. Therefore, Step 4 is correct. 

Since steps 1-4 exhaust all possible scenarios, all the invariants are preserved and the \Bslfrew{} operation is correct.

Before moving on, we need to analyze the time complexity of these new operations.

\subsection{The Analysis of Extended Operations} 
\label{subsec:BSLAnalysisExtOps}

We now analyze the time complexity of the extended operations described above using the potential method.

The extended operations, \Bslfspl{}, \Bslfjoin{}, and \Bslfrew{}, are executed on the elements of a set of biased skip lists. Let \indatast k be the set of biased skip lists after the \kth\ operation, where \indatast 0 is the initial set of biased skip lists we are given. 

Let \plats k be the set of plateaus which contain an element in the cover of some biased skip list in \indatast{k}. For a plateau $p$, let $|p|$ be the number of nodes contained in plateau $p$. We define the following sets of plateaus. Let $\plasetclosed{x}{y}{k} = \{p\in \plats k\,|\, x\leq |p|\leq y \}$ and $\plasetopen{x}{y}{k} = \{p\in \plats k\,|\, x< |p|< y \}$. 
We can now define a potential function as follows. 
\[ 
\locpotfun{\indatast k} = \cnsg\cdot(8|\plasetsmall k| + 4|\plaseta k| + |\plasetb k| + 3|\plasetfat k| + 5|\plasetobese k|)
\] 
where \cnsg\ is some constant. Observe that $\locpotfun{\indatast k}\geq 0$ for any $k$. 
For each extended operation, we will use this potential function to prove upper bounds on the amortized time complexity of the operation as well as worst-case time complexity of a sequence of operations.

\begin{lemma} 
\label{lem:costofbslfsplit} 
The $(\bsl A, \bsl B) \leftarrow \Bslfsplx{f}$ operation, where $\node f\in \bsl S$, has an amortized time complexity of 
\[ {\cal O}(\min(\intheight{\beforeop{A}},\intheight{\beforeop{B}}) - \min(\noderank{\intmax{\beforeop{A}}}, \noderank{\intmin{\beforeop{B}}})+1)\] 
where $\beforeop{A} = \prebsl{S}{f}$ and $\beforeop{B}=\postbsl{S}{f}$. 
\end{lemma}

\begin{proof} 
Let \Bslfsplx{f} be the \kth\ operation. Observe that Step 1 takes $\mathcal O(\min(\intheight{\beforeop{A}},\intheight{\beforeop{B}}) - \min(\nodeheight{\intmax{\beforeop{A}}}, \nodeheight{\intmin{\beforeop{B}}})+1)$ time. 
Step 2 takes constant time at each level from level $\min(\noderank{\intmax{\beforeop{A}}}, \noderank{\intmin{\beforeop{B}}})$ to level $\min(\intheight{\beforeop{A}},\intheight{\beforeop{B}})$. We will show that the time spent by Step~2 on levels greater than $\min(\intheight{\beforeop{A}},\intheight{\beforeop{B}})$ is essentially negligible by showing that it equals the decrease in potential at these levels. 

We now analyze the contribution of Step 1 and Step 2 to the potential change, $\locpotfun{\indatast k} - \locpotfun{\indatast{k-1}}$. 

\paragraph{Step 1}

Due to the split, the plateaus on the cover of \bsl S become plateaus on the covers of \bsl A and \bsl B. Additionally, the plateaus on levels less than or equal to $\min(\intheight{\beforeop{A}}, \intheight{\beforeop{B}})$ on the right cover of \bsl A and left cover of \bsl B are added to \plats k. 
Therefore, the contribution of Step 1 to the potential change is at most $\mathcal O(\min(\intheight{\beforeop{A}},\intheight{\beforeop{B}}) - \min(\nodeheight{\intmax{\beforeop{A}}},\nodeheight{\intmin{\beforeop{B}}})+1)$. 

\paragraph{Step 2} 
We now look at the contribution of Step 2 to the potential change, $\locpotfun{\indatast k} - \locpotfun{\indatast{k-1}}$. Note that at any level less than or equal to $\min(\intheight{\beforeop{A}}, \intheight{\beforeop{B}})$, the potential increase can be at most a constant. Therefore, the maximum contribution of Step~2 to the potential change in these levels is $\mathcal O(\min(\intheight{\beforeop{A}}, \intheight{\beforeop{B}}) - \min(\noderank{\intmax{\beforeop{A}}},\noderank{\intmin{\beforeop{B}}}))$. Next, we bound the contribution of Step~2 to the potential change in levels greater than $\min(\intheight{\beforeop{A}}, \intheight{\beforeop{B}})$. 

\paragraph{Demotions} 
Consider a demotion operation to restore \invtwo{} at some node \node x. This demotion could cause a change in potential in two ways. 

First, the plateau $p'$ which was causing the \invtwo{} violation could have had less than $a$ nodes, and now has more. 

Note that since we demote node \node x to the level of $p'$, there must be a plateau of at least $a$ nodes (due to \invtwo{}) on the other side of node \node x. Therefore, $p'$ will have at least $a+1$ nodes. Also, the plateau on the other side of node \node x cannot have more than $b$ nodes. Therefore, $p'$ will have at most $b+a$ nodes. This implies that $p'$ can only  have between $a+1$ and $b+a$ nodes. 
If $p'$ has between $a+1$ and $b-1$ nodes, the contribution of this part to the potential change is $-8\cnsg$. 
If $p'$ has $b$ nodes, the contribution of this part to the potential change is $-8\cnsg + \cnsg = -7\cnsg$. 
If $p'$ has between $b+1$ and $b+a$ nodes, the contribution of this part to the potential change is $-8\cnsg + 3\cnsg = -5\cnsg$. 
Therefore, the maximum contribution of this part to the potential change is $-5\cnsg$.

Note that $p'$ might not exist (have zero nodes) prior to the demotion of $x$. In this case, the maximum contribution of this part to the potential change is $3\cnsg$.  However, this case is only possible if $p'$ has height less than or equal to $\min(\intheight{\beforeop{A}},\intheight{\beforeop{B}})$.

Second, the demotion of $x$ could cause the plateau $p''$, which $x$ was a part of before the demotion, to have less nodes. 
If $p''$ had $2b$ nodes or more and now has between $b+1$ and $2b-1$ nodes, the contribution of this part to the potential change is $-5\cnsg + 3\cnsg = -2\cnsg$. 
If $p''$ had between $b+1$ and $2b-1$ nodes and now has $b$ nodes, the contribution of this part to the potential change is $-3\cnsg + \cnsg = -2\cnsg$. 
If $p''$ had $b$ nodes and now has between $a+1$ and $b-1$ nodes, the contribution of this part to the potential change is $-\cnsg$. 
If $p''$ had between $a+1$ and $b-1$ nodes and now has $a$ nodes, the contribution of this part to the potential change is $4\cnsg$. 
If $p''$ had $a$ nodes and now has between $1$ and $a-1$ nodes, the contribution of this part to the potential change is $-4\cnsg + 8\cnsg = 4\cnsg$. 
Additionally, if $p''$ had between $1$ and $a-1$ nodes and now has zero nodes, the contribution of this part to the potential change is $-8\cnsg$. 

Therefore, combining the first and second part, the maximum contribution of a demotion to the potential change is $7\cnsg$ on levels less than or equal to $\min(\intheight{\beforeop{A}},\intheight{\beforeop{B}})$, and $-\cnsg$ on levels greater. 

Note that the demotion of \node x could cause a high number of plateaus which did not have any elements in the cover of their biased skip list to now have an element in the cover and thus enter \plats k. In order for this case to occur, the height of \node x must be demoted at least two levels. By the description of Step~2, this implies there were no nodes of height $\nodeheight x - 1$ in the biased skip list containing \node x after Step~1. Since \bsl S has no \invtwo{} violations before Step~1, this implies there must be nodes of height $\nodeheight x - 1$ in the other biased skip list. Therefore, this case is only possible in levels less than or equal to $\min(\intheight{\beforeop{A}},\intheight{\beforeop{B}})$. 

\paragraph{Promotions} 
Consider a promotion operation at a node \node x to restore an \invone{} violation on the plateau $p'$ node \node x is on. This promotion could cause a change in potential in two ways. 

First, the plateau $p'$ could have had between $b+1$ and $2b+1$ nodes, and now has less. 
If $p'$ had between $b+1$ and $2b-1$ nodes, then the promotion of node \node x splits $p'$ into two plateaus. Only one of these plateaus remain in the cover of the biased skip list unless nodes of $p'$ were at the highest level of the biased skip list. 
The plateau remaining in the cover must have size greater than $a$ and less than $b$. Therefore, $p'$ will now have between $a+1$ and $b-1$ nodes, and the contribution of this part to the potential change is $-3\cnsg$. 
In the special case of both plateaus remaining in the cover, then both of them will have between $a+1$ and $b-1$ nodes, and the contribution of this part to the potential change is $-3\cnsg$. 
If $p'$ has $2b$ or more nodes, then the promotion of \node x splits $p'$ into two plateaus. Only one of these plateaus remain in the cover of the biased skip list unless nodes of $p'$ were at the highest level of the biased skip list. 
The plateau remaining in the cover must have size either $b-1$ or $b$. 
If it has $b-1$ nodes, the contribution of this part to the potential change is $-5\cnsg$. 
If it has $b$ nodes, the contribution of this part to the potential change is $-5\cnsg + \cnsg = -4\cnsg$. 
In the special case of both plateaus remaining in the cover, then either one of them has $b-1$ nodes, and the other one  has $b$ nodes, and the contribution of this part to the potential change is $-5\cnsg + \cnsg = -4\cnsg$; or both of them have $b$ nodes, and the contribution of this part to the potential change is $-5\cnsg + 2\cnsg = -3\cnsg$.

Second, the promotion of $x$ could cause the plateau $p''$, which $x$ becomes a part of after the promotion, to have more nodes. 
If $p''$ had more than $2b-1$ nodes, note that it must have had at most $2b$ nodes. Promotion of $x$ increases the size of $p''$ to $2b+1$. Thus, it does not change its set and the contribution of this part to the potential change is $0$. 
Note that, in general, if the promotion does not change the set of $p''$, then the contribution of this part to the potential change is $0$. 
If $p''$ had\hide{between $b+1$ and} $2b-1$ nodes and now has $2b$  nodes, the contribution of this part to the potential change is $-3\cnsg + 5\cnsg = 2\cnsg$. 
If $p''$ had $b$ nodes and now has $b+1$ nodes, the contribution of this part to the potential change is $-\cnsg + 3\cnsg = 2\cnsg$. 
If $p''$ had\hide{ between $a+1$ and} $b-1$ nodes and now has $b$ nodes, the contribution of this part to the potential change is $\cnsg$. 
If $p''$ had $a$ nodes and now has $a+1$ nodes, the contribution of this part to the potential change is $-4\cnsg$. 
If $p''$ had\hide{ between $1$ and} $a-1$ nodes and now has $a$ nodes, the contribution of this part to the potential change is $-8\cnsg + 4\cnsg = -4\cnsg$. 
Additionally, if $p''$ had zero nodes and now has $1$ node, the contribution of this part to the potential change is $8\cnsg$. However, this can only happen if an earlier demotion caused the only node of $p''$ to be demoted. Therefore, $p''$ first has $1$ node; then a demotion causes $p''$ to have $0$ nodes; and then a promotion associated with that demotion causes $p''$ to have $1$ node again. The effect on the potential change is zero. 

Therefore, combining the first and second part, the maximum contribution of a promotion to the potential change is $-\cnsg$. 

Let constant \cnsh\ be an upper bound on the worst-case running time of any single demotion or promotion operation. Let \amcost k and \actcost k respectively be the amortized and worst-case running time of \Bslfsplx{f} and let \vionum k be the number of violations that are restored during Step~2 above level $\min(\intheight{\beforeop{A}},\intheight{\beforeop{B}})$. 
Then we have $\actcost k \leq {\cal O}(\min(\intheight{\beforeop{A}},\intheight{\beforeop{B}}) - \min(\noderank{\intmax{\beforeop{A}}}, \noderank{\intmin{\beforeop{B}}})+1) + \cnsh\cdot \vionum k$. 
Combining the bounds on the maximum contribution of Step~1 and Step~2 to the potential change and setting $\cnsg = \cnsh$, we have $\locpotfun{\indatast k} - \locpotfun{\indatast{k-1}} \leq -\vionum k \cnsh + {\cal O}(\min(\intheight{\beforeop{A}},\intheight{\beforeop{B}}) - \min(\noderank{\intmax{\beforeop{A}}}, \noderank{\intmin{\beforeop{B}}}))$. 
This yields 
\begin{align*} 
\amcost{k} &= \actcost k + \locpotfun{\indatast k} - \locpotfun{\indatast{k-1}}\\ 
&\leq \cnsh\vionum k + \locpotfun{\indatast k} - \locpotfun{\indatast{k-1}}+ {\cal O}(\min(\intheight{\beforeop{A}},\intheight{\beforeop{B}}) - \min(\noderank{\intmax{\beforeop{A}}}, \noderank{\intmin{\beforeop{B}}})+1)\\ 
&\leq \cnsh\vionum k - \cnsh\vionum k +  {\cal O}(\min(\intheight{\beforeop{A}},\intheight{\beforeop{B}}) - \min(\noderank{\intmax{\beforeop{A}}}, \noderank{\intmin{\beforeop{B}}})+1) \\ 
&= {\cal O}(\min(\intheight{\beforeop{A}},\intheight{\beforeop{B}}) - \min(\noderank{\intmax{\beforeop{A}}}, \noderank{\intmin{\beforeop{B}}})+1). 
\end{align*} 
This completes the proof. 
\end{proof}

\begin{lemma} 
\label{lem:wccostofbslfsplit} 
Given a set $\indatast 0=\{\initbsl 1, \initbsl 2,\ldots, \initbsl m\}$ of biased skip lists, a sequence of $(S_k, T_k) \leftarrow \Bslfsplx{f_k}$ operations, for $1\leq k \leq t$ where $f_k\in \set U_k$  and $\set U_k \in \indatast{k-1}$, can be executed in worst-case time 
\[ 
\mathcal O\left(\sum_{k=1}^t (\min(\intheight{\beforeop{S_k}},\intheight{\beforeop{T_k}}) - \min(\noderank{\intmax{\beforeop{S_k}}},\noderank{\intmin{\beforeop{T_k}}})+1)\right)+
\] 
\[\mathcal O\left(\sum_{i=1}^{m}(\intheight{\initbsl i} - \min(\nodeheight{\intmin{\initbsl i}},\nodeheight{\intmax{\initbsl i}})+1)\right) 
\] 
where $\beforeop{S_k} = \prebsl{{U_k}}{f_k}$ and $\beforeop{T_k}=\postbsl{{U_k}}{f_k}$. 
\end{lemma}

\begin{proof} 
Let \amcost k and \actcost k respectively be the amortized and worst-case running time of \Bslfsplx{f_k}. Note that $\locpotfun{\indatast 0} = \mathcal O(|\plats 0|) = \mathcal O\left(\sum_{i=1}^{m}(\intheight{\initbsl i} - \min(\nodeheight{\intmin{\initbsl i}},\nodeheight{\intmax{\initbsl i}})+1)\right)$. Then we have 
{\allowdisplaybreaks 
\begin{align*} 
\sum_{k=1}^{t} \amcost{k} &= \sum_{k=1}^{t} (\actcost k + \locpotfun{\indatast k} - \locpotfun{\indatast{k-1}})\\ 
\sum_{k=1}^{t} \amcost{k} &=\sum_{k=1}^{t} \actcost k + \locpotfun{\indatast{t}} - \locpotfun{\indatast{0}}\\ 
\sum_{k=1}^{t} \actcost k &=\sum_{k=1}^{t} \amcost{k} - \locpotfun{\indatast{t}}+ \locpotfun{\indatast{0}}\\ 
\sum_{k=1}^{t} \actcost k &\leq\sum_{k=1}^{t} \amcost{k} + \locpotfun{\indatast{0}}\\ 
&= \mathcal O\left(\sum_{k=1}^t (\min(\intheight{\beforeop{S_k}},\intheight{\beforeop{T_k}}) - \min(\noderank{\intmax{\beforeop{S_k}}},\noderank{\intmin{\beforeop{T_k}}})+1)\right)&& \mbox{By Lemma}~\ref{lem:costofbslfsplit}\\ 
& \hspace{4mm} + \locpotfun{\indatast{0}}\\ 
&= \mathcal O\left(\sum_{k=1}^t (\min(\intheight{\beforeop{S_k}},\intheight{\beforeop{T_k}}) - \min(\noderank{\intmax{\beforeop{S_k}}},\noderank{\intmin{\beforeop{T_k}}})+1)\right)\\ 
& \hspace{4mm}  +\mathcal O\left(\sum_{i=1}^{m}(\intheight{\initbsl i} - \min(\nodeheight{\intmin{\initbsl i}},\nodeheight{\intmax{\initbsl i}})+1)\right). 
\end{align*} 
} 
\end{proof}

\begin{lemma} 
\label{lem:costofbslfjoin} 
The $\bsl S \leftarrow \Bslfjoinx{\ell}{r}$ operation, where $\ell\in A$, $r\in B$, has an amortized time complexity of 
\[{\cal O}\left(\min(\intheight{A},\intheight{B}) - \min(\nodeheight{\intmax{A}}, \nodeheight{\intmin{B}})+1\right).\]  
\end{lemma}

\begin{proof} 
Let \Bslfjoinx{\ell}{r} be the \kth\ operation. Observe that Step 1 takes $\mathcal O(\min(\intheight{A},\intheight{B}) - \min(\nodeheight{\intmax{A}},\nodeheight{\intmin{B}})+1)$ time. 
Then, Step 2 takes constant time at each level from level $\min(\nodeheight{\intmax{A}}, \nodeheight{\intmin{B}})$ to level $\min(\intheight{A},\intheight{B})$. We will show that the time spent by Step~2 on levels greater than $\min(\intheight{A},\intheight{B})$ is essentially negligible by showing that it equals the decrease in potential at these levels. 

We now analyze the contribution of Step 1 and Step 2 to the potential change, $\locpotfun{\indatast k} - \locpotfun{\indatast{k-1}}$. 

\paragraph{Step 1} 

Due to the join, the plateaus on levels greater than $\min(\intheight A, \intheight B)$ as well as the plateaus on left cover of \bsl A and right cover of \bsl B becomes plateaus on the cover of \bsl S. The plateaus on levels less than or equals to $\min(\intheight A, \intheight B)$ on the right cover of \bsl A and on the left cover of \bsl B are not on the cover of \bsl S and thus are not added to \plats k. 

The only way Step~1 can increase the potential is if all the plateaus on the covers of \bsl A and \bsl B becomes plateaus on the cover of \bsl S and the plateaus containing nodes \intmax A and \intmin B merge and form a larger plateau with more weight with respect to the potential function. 
Therefore, the contribution of Step 1 to the potential change is at most $3\cnsg$.

\paragraph{Step 2} 

Note that Step 2 only involves promotions. Any promotion on level less than or equal to $\min(\intheight A, \intheight B)$ will increase the potential by at most a constant and total contribution to the potential change is bounded by $\mathcal O(\min(\intheight{A},\intheight{B}) - \min(\nodeheight{\intmax{A}}, \nodeheight{\intmin{B}})+1)$. We will focus on promotions which occur on higher levels. 

Consider a promotion operation at a node $x$ to restore an \invone{} violation on the plateau $p'$ node \node x is on where $p'$ is on a level greater than $\min(\intheight A, \intheight B)$. The promotion could cause a change in potential in two ways.

First, the plateau $p'$ could have had $b+1$ nodes, then the promotion of node \node x splits $p'$ into two plateaus with between $a+1$ and $b-1$ nodes each. Then, the contribution of this part to the potential change is $-3\cnsg$.

Second, the promotion of $x$ could cause the plateau $p''$, which $x$ becomes a part of after the promotion, to have more nodes. Note that due to \invone{}, $p''$ could not have had more than $b$ nodes. If $p''$ had $b$ nodes and now has $b+1$ nodes, the contribution of this part to the potential change is $-\cnsg + 3\cnsg = 2\cnsg$. If $p''$ had less than $b$ nodes and now has at most $b$ nodes, the potential could increase by at most a constant, but the promotion will not cause an \invone{} violation at $p''$, and Step 2 will terminate.

Therefore, combining the first and second part, the maximum contribution of a promotion to the potential change is $-\cnsg$ except possibly for the last promotion.

Recall that the worst-case running time of any single demotion or promotion operation is bounded by constant \cnsh. Let \amcost k and \actcost k respectively be the amortized and worst-case running time of \Bslfjoinx{\ell}{r} and let \vionum k be the number of violations that are restored during Step~2 above level $\min(\intheight{A},\intheight{B})$. 
Then we have $\actcost k \leq {\cal O}(\min(\intheight{A},\intheight{B}) - \min(\nodeheight{\intmax{A}}, \nodeheight{\intmin{B}})+1) + \cnsj\cdot \vionum k$. 
Combining the bounds on the maximum contribution of Step~1 and Step~2 to the potential change and setting $\cnsg = \cnsj$, we have $\locpotfun{\indatast k} - \locpotfun{\indatast{k-1}} \leq -\vionum k \cnsj + {\cal O}(\min(\intheight{A},\intheight{B}) - \min(\nodeheight{\intmax{A}}, \nodeheight{\intmin{B}})+1)$. 
This yields 

\begin{align*} 
\amcost{k} &= \actcost k + \locpotfun{\indatast k} - \locpotfun{\indatast{k-1}}\\ 
&\leq \cnsj\vionum k + \locpotfun{\indatast k} - \locpotfun{\indatast{k-1}}+ {\cal O}(\min(\intheight{A},\intheight{B}) - \min(\nodeheight{\intmax{A}}, \nodeheight{\intmin{B}})+1)\\ 
&\leq \cnsj\vionum k - \cnsj\vionum k +  {\cal O}(\min(\intheight{A},\intheight{B}) - \min(\nodeheight{\intmax{A}}, \nodeheight{\intmin{B}})+1) \\ 
&= {\cal O}(\min(\intheight{A},\intheight{B}) - \min(\nodeheight{\intmax{A}}, \nodeheight{\intmin{B}})+1). 
\end{align*} 
This completes the proof. 
\end{proof}

\begin{lemma} 
\label{lem:wccostofbslfjoin} 

Given a set $\indatast 0=\{\initbsl 1, \initbsl 2,\ldots, \initbsl m\}$ of biased skip lists, and a sequence of $U_k \leftarrow \Bslfjoinx{\ell_k}{r_k}$ operations, for $1\leq k \leq t$ where $\ell_k\in \set S_k$, $r_k\in \set T_k$  and $\set S_k,\set T_k \in \indatast{k-1}$, can be executed in worst-case time 
\[ 
\mathcal O\left(\sum_{k=1}^t (\min(\intheight{S_k},\intheight{T_k}) - \min(\nodeheight{\intmax{S_k}},\nodeheight{\intmin{T_k}})+1)\right)+
\] 
\[\mathcal O\left(\sum_{i=1}^{m}(\intheight{\initbsl i} - \min(\nodeheight{\intmin{\initbsl i}},\nodeheight{\intmax{\initbsl i}})+1)\right). 
\] 
\end{lemma}

\begin{proof} 
Let \amcost k and \actcost k respectively be the amortized and worst-case running time of \Bslfjoinx{\ell_k}{r_k}. Note that $\locpotfun{\indatast 0} = \mathcal O(|\plats 0|) = \mathcal O\left(\sum_{i=1}^{m}(\intheight{\initbsl i} - \min(\nodeheight{\intmin{\initbsl i}},\nodeheight{\intmax{\initbsl i}})+1)\right)$. Then we have 

{\allowdisplaybreaks 
\begin{align*} 
\sum_{k=1}^{t} \amcost{k} &= \sum_{k=1}^{t} (\actcost k + \locpotfun{\indatast k} - \locpotfun{\indatast{k-1}})\\ 
\sum_{k=1}^{t} \amcost{k} &=\sum_{k=1}^{t} \actcost k + \locpotfun{\indatast{t}} - \locpotfun{\indatast{0}}\\ 
\sum_{k=1}^{t} \actcost k &=\sum_{k=1}^{t} \amcost{k} - \locpotfun{\indatast{t}}+ \locpotfun{\indatast{0}}\\ 
\sum_{k=1}^{t} \actcost k &\leq\sum_{k=1}^{t} \amcost{k} + \locpotfun{\indatast{0}}\\ 
&= \mathcal O\left(\sum_{k=1}^t (\min(\intheight{S_k},\intheight{T_k}) - \min(\nodeheight{\intmax{S_k}},\nodeheight{\intmin{T_k}})+1)\right)&& \mbox{By Lemma}~\ref{lem:costofbslfjoin}\\ 
& \hspace{4mm} + \locpotfun{\indatast{0}}\\ 
&= \mathcal O\left(\sum_{k=1}^t (\min(\intheight{S_k},\intheight{T_k}) - \min(\nodeheight{\intmax{S_k}},\nodeheight{\intmin{T_k}})+1)\right)\\ 
& \hspace{4mm}  +\mathcal O\left(\sum_{i=1}^{m}(\intheight{\initbsl i} - \min(\nodeheight{\intmin{\initbsl i}},\nodeheight{\intmax{\initbsl i}})+1)\right).
\end{align*} 
} 
\end{proof}

\begin{lemma} 
\label{lem:costofbslfreweight} 
The $\Bslfrewx{f}{w}$ operation, where $f\in S$, has a worst-case and amortized time complexity of 
\[{\cal O}(\max(\intheight{S},\nodenewrank f) - \min(\nodeheight f, \nodenewrank f)+1).\] 
\end{lemma}

\begin{proof} 
\Bslfrewx{f}{w} only spends constant time at each level between $\min(\nodeheight f, \nodenewrank f)$ and $\max(\nodeheight f, \nodenewrank f)$ for promoting or demoting \node f; and at most constant time at each level between $\min(\nodeheight f, \nodenewrank f)$ and $\max(\intheight{S},\nodenewrank f)$ for restoring invariants. Therefore, the worst-case complexity of \Bslfrewx{f}{w} is ${\cal O}(\max(\intheight{S},\nodenewrank f) - \min(\nodeheight f, \nodenewrank f)+1)$. 

Let \Bslfrewx{f}{w} be the \kth\ operation. Note that no plateaus that are on levels less than $\min(\nodeheight f, \nodenewrank f)$ are affected by the operation. Since the potential increase associated with each level is bounded by a constant, we have $\locpotfun{\indatast{k}} - \locpotfun{\indatast{k-1}} = \mathcal O(\max(\intheight{S},\nodenewrank f) - \min(\nodeheight f, \nodenewrank f))$. Thus, the lemma follows. 
\end{proof}

\begin{lemma} 
\label{lem:logncostofextended} 
The \Bslfsrc{}, \Bslfspl{}, \Bslfjoin{}, and \Bslfrew{} operations all have a worst-case time complexity of ${\cal O}(\log n)$. 
\end{lemma}

\begin{proof} 
Let $W'$ be the sum of the weights of all elements in the sets involved in any of these four operations. By our previous assumptions, $W'$ is bounded from above by a polynomial in $n$. Since these versions of the operations are more efficient than the non-finger versions which take ${\cal O}(\log W')$ time in the worst case, these operations have a worst-case time complexity of  ${\cal O}(\log n)$. 
\end{proof}

We now present our data structure, the \Ds{}, which improves the worst-case bound in Theorem~\ref{thm:logsq} by a factor of $\log n$, matching the lower bound of \cite{conf/stoc/PatrascuD04}.

\section{Our Data Structure: The Mergeable Dictionary} 
\label{sec:DS}

The \Ds{} stores each set in the collection \collection S as a biased skip list.  The weight of each node in each biased skip list is determined by \collection S. When the collection of sets is modified, for instance via a \Union{} operation, in order to reflect this change in the data structure, besides splitting and joining biased skip lists, we need to ensure the weights of the affected nodes are properly updated and biased skip list invariants \invzero{}, \invone{}, and \invtwo{} are preserved. For simplicity we assume that \datast 0 is the collection of singleton sets and \datast i, for all $i$, partitions the universe \universe. This lets us precompute, for each node \node x, \nodeposition{\universe}{x}, the global position of \node x. 
  
For the \Union{} algorithm, we will use the same basic approach outlined in Section~\ref{sec:LogSQ}, the \segment{} merging heuristic, which works by extracting the \segments{} from each set and then gluing them together to form the union of the two sets. 

As previously mentioned, while we do not have control over the number of \segments{} that need to be processed which has been shown to have an amortized lower bound of $\Omega(\log n)$ per operation, we need to process each \segment{} faster. In order to do so, we depart from balanced search trees and instead use Biased Skip Lists \cite{journals/algorithmica/BagchiBG05} with the extended operations we introduced in Section~\ref{sec:BSL}.

Before we discuss the implementation of each operation in detail, we need to describe the weighting scheme. 

\subsection{Weighting Scheme} 
\label{subsec:DSWeights} 
Let the weight of a node $x$, $\nodeweight x$, be the sum of the sizes of its adjacent gaps. In other words, if $\nodeposition{S}{x} = k$ for some node $x\in S$, then we have 
\[ 
\label{eqn:weight} 
\nodeweight x = \numgap{S}{k-1} +\numgap{S}{k}. 
\] 
Recall that $\numgap{S}{0} = \numgap{S}{|\set S|} = 1$. Observe that this implies for any set \set S, $\setweight S \leq 2n$. 

\subsection{The \hide{\Ms{}, }\Find{}, \Src{}, and \Spl{} Operations} 
\label{subsec:DSOthers}

The \Findx{x} operation can simply return the maximum element of the set which contains $x$ by invoking \Bslfsrcx{i}{+\infty}. The \Srcx{X}{i} operation can be performed by simply invoking \Bslsrcx{X}{i}. The \Splx{X}{i} operation can be performed by simply invoking \Bslsplx{X}{i} and running \Bslfrew{} on one node in each of the resulting biased skip lists to restore the weights.

\subsection{The \Union{} Operation} 
\label{subsec:DSUnion} 
We will use the \Ds{} in Figure \ref{fig:example} to illustrate the \Union{} operation. \Unionx{A}{B} operation can be viewed as having four essential phases: finding the \segments{}, extracting the \segments{}, updating the weights, and gluing the \segments{}. 
A more detailed description follows.

\paragraph{Phase I: Finding the \segments{}} 
Assume $\intmin{A} < \intmin{B}$ w.l.o.g. Let $\lasta = \lceil\intnum{A}{B}/2\rceil$ and $\lastb = \lfloor\intnum{A}{B}/2\rfloor$. Recall that $\intsubs{A}{B} = \{\intl A_1, \intl B_1, \intl A_2, \intl B_2, \ldots\}$ is the set of \segments{} associated with the \Unionx{A}{B} operation where $\intl A_i$ and $\intl B_i$ are the \ith\ \segment{} of $A$ and $B$ respectively. We have $\intmin{A_1} = \intmin{A}$ and $\intmin{B_1} = \intmin{B}$. Given \intmin{A_i} and \intmin{B_i}, we find \intmax{A_i} by invoking \Bslfsrcx{\intmin{A_i}}{\intmin{B_i}}. Similarly, given \intmin{B_i} and \intmin{A_{i+1}} we find \intmax{B_i} by invoking \Bslfsrcx{\intmin{B_i}}{\intmin{A_{i+1}}}. Lastly, given \intmax{A_i} and \intmax{B_i}, observe that $\intmin{A_{i+1}} = \suc{A}{\intmax{A_i}}$ and $\intmin{B_{i+1}} = \suc{B}{\intmax{B_i}}$. 
Note that the \suc{}{} operation is performed in constant time in a biased skip list using the lowest successor link of a node. 
At the end of this phase, all the \segments{} are found (see Figure \ref{fig:findints}). Specifically, we have computed for all $i$ and $j$ (\intmin{A_i}, \intmax{A_i}) and (\intmin{B_j}, \intmax{B_j}). 

\paragraph{Phase II: Extracting the \segments{}} 

Since we know where the minimum and maximum node of each \segment{} is from the previous phase, we can extract all the \segments{} easily in order by invoking \Bslfsplx{\intmax{A_i}} for $1\leq i< \lasta$, and \Bslfsplx{\intmax{B_j}} for $1\leq j< \lastb$ (see Figure \ref{fig:extractints}).

\begin{figure}[H]
\centering 
\includegraphics[scale=1]{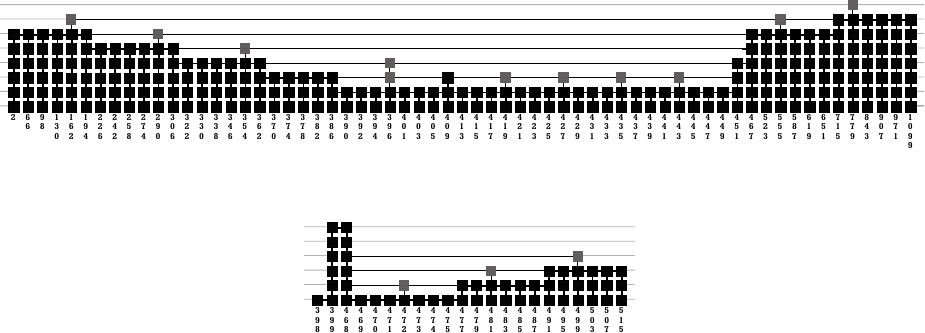} 
\caption{We have two biased skip lists \bsl A (top) and \bsl B (bottom). The tower of each node is represented by a set of vertically stacked squares which represent individual levels of a node. The predecessor/successor pointers are also shown. The lightly shaded levels exist due to promotions. The position of each node \node x, \nodeposition{\universe}{x}, is indicated at the bottom of the tower of each node.} 
\label{fig:example} 
\end{figure}
\begin{figure}[H]
\centering 
\includegraphics[scale=1]{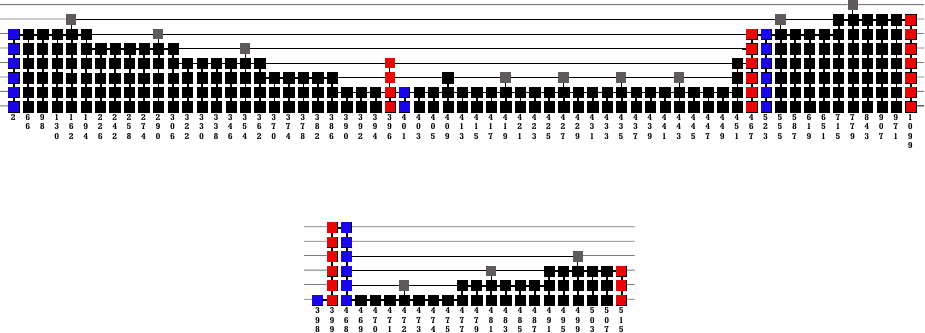} 
\caption{At the end of Phase~I, we have all the minimum and maximum nodes of the \segments{} of the \Unionx{A}{B} operation. The minimum nodes are colored blue, the maximum nodes are colored red.} 
\label{fig:findints} 
\end{figure}

\paragraph{Phase III: Updating Weights} Next, we need to update the weights of the affected nodes. Let the new weight of item $x$ be $\nodenewweight x$. Then

{\allowdisplaybreaks 
\begin{align*} 
\mbox{For } 2\leq i &\leq \lasta, \mbox{ let}\\ 
&\nodenewweight{\intmin{A_i}} = \nodeweight{\intmin{A_i}} + \nodeposition{U}{\intmax{A_{i-1}}} - \nodeposition{U}{\intmax{B_{i-1}}},\\ 
\mbox{For } 2\leq i &\leq \lastb , \mbox{ let}\\ 
&\nodenewweight{\intmin{B_i}} = \nodeweight{\intmin{B_i}} + \nodeposition{U}{\intmax{B_{i-1}}} - \nodeposition{U}{\intmax{A_{i}}}, \\ 
\mbox{For } 1\leq i &< \lasta, \mbox{ let}\\ 
&\nodenewweight{\intmax{A_i}} = \nodeweight{\intmax{A_i}} + \nodeposition{U}{\intmin{B_i}} - \nodeposition{U}{\intmin{A_{i+1}}}, \\ 
\mbox{For } 1\leq i &< \lastb , \mbox{ let}\\ 
&\nodenewweight{\intmax{B_i}} = \nodeweight{\intmax{B_i}} + \nodeposition{U}{\intmin{A_{i+1}}} - \nodeposition{U}{\intmin{B_{i+1}}}.
\end{align*}
}
We also have 
\[\nodenewweight{\intmin{B_1}} = \nodeweight{\intmin{B_1}} -1+ \nodeposition{U}{\intmin{B_{1}}} - \nodeposition{U}{\intmax{A_{1}}}. \] 
If $\intnum{A}{B}$ is even, we have 
\[\nodenewweight{\intmax{A_\lasta}} = \nodeweight{\intmax{A_\lasta}} -1 + \nodeposition{U}{\intmin{B_\lasta}} - \nodeposition{U}{\intmax{A_{\lasta}}}. \] 
If $\intnum{A}{B}$ is odd, we have 
\[\nodenewweight{\intmax{B_{\lastb}}} = \nodeweight{\intmax{B_{\lastb}}} -1+ \nodeposition{U}{\intmin{A_{\lasta}}} - \nodeposition{U}{\intmax{B_{\lastb}}}.\]  We can perform these weight updates (see Figure \ref{fig:updateweights}) by invoking

\begin{itemize} 
\item \Bslfrewx{\intmin{A_i}}{\nodenewweight{\intmin{A_i}}} \quad for  $2\leq i \leq \lasta$, 
\item \Bslfrewx{\intmax{A_i}}{\nodenewweight{\intmax{A_i}}} \quad for  $1\leq i \leq \lastb$, 
\item \Bslfrewx{\intmin{B_j}}{\nodenewweight{\intmin{B_j}}} \quad for  $1\leq j \leq \lastb$, 
\item \Bslfrewx{\intmax{B_j}}{\nodenewweight{\intmax{B_j}}} \quad for  $1\leq j < \lasta$. 
\end{itemize}

\paragraph{Phase IV: Gluing the \segments{}} Since we assumed w.l.o.g.~that $\intmin A < \intmin B$, the correct order of the \segments{} is $(\intl A_1, \intl B_1, \intl A_2, \intl B_2, \ldots)$ by construction. We can glue all the \segments{} by invoking \Bslfjoinx{\intmax{A_{i}}}{\intmin{B_i}} for $1\leq i \leq \lastb$ and \Bslfjoinx{\intmax{B_{i}}}{\intmin{A_{i+1}}} for $1\leq i < \lasta$ (see Figure \ref{fig:glueints}). 

This concludes the presentation of our data structure. In the next section we will analyze the running time of each of the \opcount{} operations.

\begin{figure}[H]
\centering 
\includegraphics[scale=1]{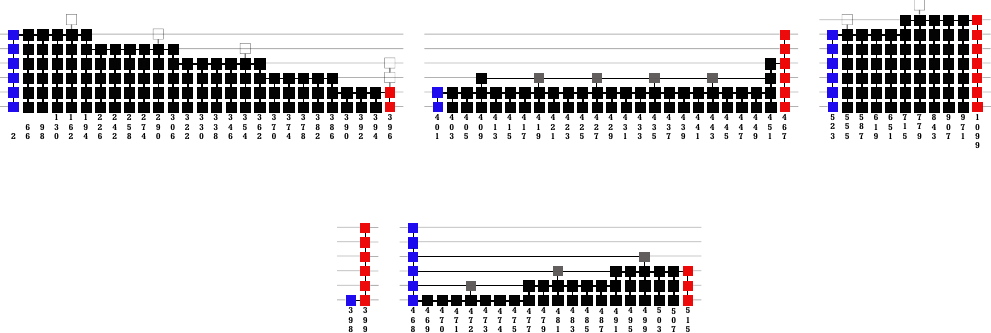} 
\caption{At the end of Phase~II, all the \segments{} are extracted. These extractions cause invariants \invone{} and \invtwo{} to be violated, which are then restored by \Bslfspl{}. Hollow levels denote demotions.} 
\label{fig:extractints} 
\end{figure}

\begin{figure}[H]
\centering 
\includegraphics[scale=1]{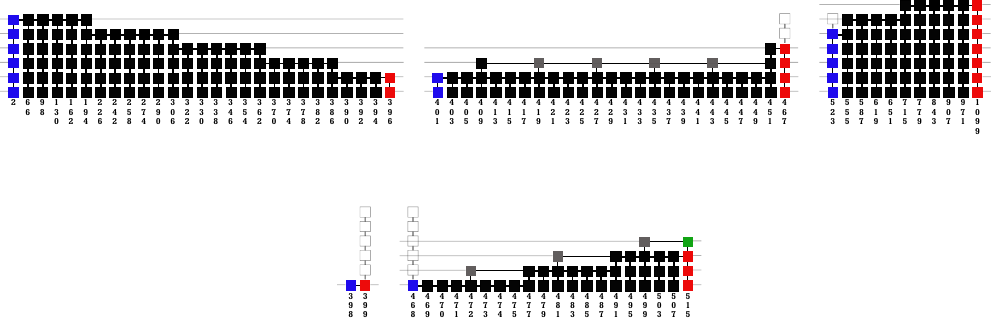} 
\caption{The weights of the minimum and maximum nodes of each \segment{} are affected by the \Unionx{A}{B} operation. Accordingly, in Phase~III, we update the weights of these nodes. This update causes \invzero{} violations. The \Bslfrew{} operation first restores \invzero{}. Restoring \invzero{} causes \invone{} and \invtwo{} violations which are then again restored by \Bslfrew{}. The hollow levels denote demotions and green levels denote promotions.} 
\label{fig:updateweights} 
\end{figure}

\begin{figure}[H]
\centering 
\includegraphics[scale=1]{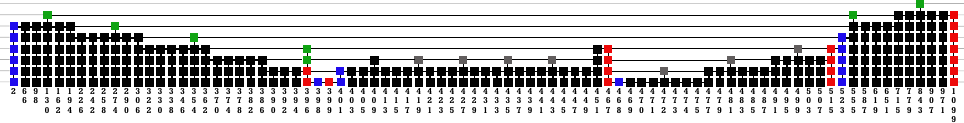} 
\caption{In Phase~IV, we join all the \segments{} to obtain the union of \set A and \set B. These joins cause invariants \invone{} and \invtwo{} to be violated, which are then restored by \Bslfjoin{}. The green levels denote promotions.} 
\label{fig:glueints} 
\end{figure}

\section{Analysis of the Mergeable Dictionary} 
\label{sec:Analysis}

Before we can analyze the amortized time complexity of the \Ds{} operations, we need a new potential function which we will present in Section~\ref{subsec:AnalysisPotential}. We then prove that all the operations except \Union{} have a worst-case and amortized time complexity of ${\cal O}(\log n)$ in Section~\ref{subsec:AnalysisOthers}. Lastly, we will show that the \Union{} operation has an amortized time complexity of ${\cal O}(\log n)$ in Section~\ref{subsec:AnalysisUnion}.

\subsection{The New Potential Function} 
\label{subsec:AnalysisPotential}

Let \datast i be the data structure containing our dynamic collection of disjoint sets, $\collection S^{(i)} = \{\set S_1^{(i)}, \set S_2^{(i)},\ldots\}$ after the \ith\ operation. Let 
\[\varphi(\set S) = \sum_{\node x\in \set S} (\log \numgap{S}{\nodeposition{S}{x}-1} + \log \numgap{S}{\nodeposition{S}{x}}).\] 
 Then we define the potential after the \ith\ operation as follows. 
\begin{equation}
\label{eq:mainpotential}
\Phi(\datast i) = \cnsd\cdot\sum_j \varphi(\set S_j^{(i)}) 
\end{equation}
where $\cnsd$ is a constant to be determined later. 

Note that the main difference between this 
function and the one in Section~\ref{subsec:LogSQPotFun} is the elimination of the $\log n$ term.

\subsection{The Analysis of the \hide{\Ms{}, }\Find{}, \Src{}, and \Spl{} Operations} 
\label{subsec:AnalysisOthers} 
We now show that all the operations except \Union{} have a worst-case time  complexity of ${\cal O}(\log n)$, and they do not cause a substantial increase in the potential which yields that  their amortized time complexity is also ${\cal O}(\log n)$. 

\begin{thm} 
\label{thm:otheropsbound} 
The worst-case and amortized time complexity of the \Findx{x}, \Srcx{S}{x}, and \Splx{S}{x} operations is ${\cal O}(\log n)$. 
\end{thm} 

\begin{proof} 
The worst-case time complexity of the \Bslfsrc{} operation invoked by the \Find{} and \Src{} operations is ${\cal O}(\log n)$ by Lemma~\ref{lem:logncostofextended}.  Recall that since these operations do not change the structure, the potential remains the same. Therefore, worst-case and amortized time complexity of \Find{} and \Src{} is ${\cal O}(\log n)$. 
The worst-case time-complexity of the \Bslspl{} and \Bslrew{} operations invoked by the \Spl{} operation is ${\cal O}(\log n)$ by Lemma~\ref{lem:logncostofextended}. Observe that \Spl{} can only decrease the potential. Therefore, the worst-case and amortized time complexity of \Spl{} is ${\cal O}(\log n)$. 
\end{proof}

\subsection{The Analysis of the \Union{} Operation} 
\label{subsec:AnalysisUnion} 
All we have left to do is show that the amortized time complexity of the \Union{} operation is ${\cal O}(\log n)$.  In order to do this, first we will show that the worst-case time complexity of the \Unionx{A}{B} operation is $\mathcal O\left(\log n + \sum_i \intfun{A_i}+\sum_j \intfun{B_j}\right)$. We define \intfun{A_i} and \intfun{B_j} next. 

\begin{defn} 
Consider the \Unionx{A}{B} operation. 
Recall that \nodenewweight{x} is the new weight of node \node x after the \Unionx{A}{B} operation. Also recall that $\lasta = \lceil\intnum{A}{B}/2\rceil$ and $\lastb = \lfloor\intnum{A}{B}/2\rfloor$. 
Then, for $1< i < \lasta$, let 
\[ 
\intfun{A_i} = \log\frac{\nodeweight{\intmax{A_{i-1}}}+\nodeweight{\intmin{A_{i+1}}} + \sum_{\node x\in \intl A_i}{\nodeweight x}}{\min(\nodenewweight{\intmax{{B_{i-1}}}}, \nodenewweight{\intmin{{A_i}}}, \nodenewweight{\intmax{{A_i}}}, \nodenewweight{\intmin{{B_{i}}}})} 
\] 
and for $1< j < \lastb$, let 
\[ 
\intfun{B_j} = \log\frac{\nodeweight{\intmax{B_{j-1}}}+\nodeweight{\intmin{B_{j+1}}} + \sum_{\node x\in \intl B_j}{\nodeweight x}}{\min(\nodenewweight{\intmax{{A_{j}}}}, \nodenewweight{\intmin{{B_j}}}, \nodenewweight{\intmax{{B_j}}}, \nodenewweight{\intmin{{A_{j+1}}}})}. 
\] 
For the boundary cases, let $\intfun{A_1} = \intfun{B_1} = \intfun{A_{\lasta}} = \intfun{B_{\lastb}} = \log n$. 
\end{defn}

\subsubsection{The Worst-Case Time Complexity} 
\label{subsubsec:AnalysisUnionWorst}

We need to bound the worst-case time complexity of each phase of the \Unionx{A}{B} operation.

\begin{lemma} 
\label{lem:phaseone} 
Phase~I of the \Union{} operation has a worst-case time complexity of \[\mathcal O\left(\log n+ \sum_{i=2}^{\lasta-1} \intfun{A_i} +\sum_{j=2}^{\lastb-1} \intfun{B_j}\right).\] 
\end{lemma} 

\begin{proof} 
During Phase~I, \Bslfsrcx{\intmin{A_i}}{t} is invoked for $1\leq i < \lasta$ where $\pre{A}{t} = \intmax{A_i}$ and $\suc{A}{t} = \intmin{A_{i+1}}$;  and \Bslfsrcx{\intmin{B_j}}{s} is invoked for $1\leq j < \lastb$ where $\pre{B}{s} = \intmax{B_j}$ and $\suc{B}{s} = \intmin{B_{j+1}}$. Observe that $\nodenewweight{\intmin{{B_i}}}<\nodeweight{\intmin{A_{i+1}}}$ and $\nodenewweight{\intmin{{B_i}}}<\nodeweight{\intmax{A_{i}}}$. Therefore, by the definition of \intfun{A_i} and \intfun{B_j}, and Lemma~\ref{lem:logncostofextended}, the worst-case time complexity of Phase~I is $\mathcal O\left(\log n+ \sum_{i=2}^{\lasta-1} \intfun{A_i} +\sum_{j=2}^{\lastb-1} \intfun{B_j}\right)$. 
\end{proof} 

We will need the following lemma to bound the worst-case time complexity of Phase~II-IV. 

\begin{lemma} 
\label{lem:height} 
Given a biased skip list \bsl S and any node $\node f \in \bsl S$, recall that $\prebsl{S}{f} = \{x\in S\,|\,x\leq f\}$. Then, 
\[ 
\intheight{\prebsl{S}{f}} \leq \log \setweight{\prebsl{S}{f}}. 
\] 
\end{lemma} 

\begin{proof} 
Let $\maxrank = \max_{\node x\in \prebslsub{S}{f}} \noderank x$ and $\ranksatlevel{t} = \{x\in\prebsl{S}{f}\,|\,\noderank{x} = t\}$. Also, let $\nodesatlevel{t} = \{x\in\prebsl{S}{f}\,|\,\nodeheight{x} \geq t,\, \noderank x \leq t\}$. 
Observe that since $a=2$, due to \invtwo{}, $\nodesatlevel i \geq 2(\nodesatlevel{i+1}-\ranksatlevel{i+1})$.
This yields
\begin{align*} 
\nodesatlevel{1} 
&\geq 2(\nodesatlevel{2} - \ranksatlevel{2})\\ 
&\quad\vdots\\
&\geq 2^{j-1}\nodesatlevel{j} - \sum_{i=1}^{j-1} 2^{i-1}\ranksatlevel{1+i}\\ 
&= 2^{j-1}\nodesatlevel{j} - \sum_{i=2}^{j} 2^{i}\ranksatlevel{i}\\ 
&\geq 2^{\intheight{\prebslsub{S}{f}}-1}\nodesatlevel{\intheight{\prebslsub{S}{f}}} - \sum_{i=2}^{\intheight{\prebslsub{S}{f}}} 2^{i}\ranksatlevel{i}\\ 
&= 2^{\intheight{\prebslsub{S}{f}}-1}\nodesatlevel{\intheight{\prebslsub{S}{f}}} - \sum_{i=2}^{\maxrank} 2^{i}\ranksatlevel{i}
&& \ranksatlevel{t} = 0 \text{ for } t>R\\  
&\geq 2^{\intheight{\prebslsub{S}{f}}-1} - \sum_{i=2}^{\maxrank} 2^{i}\ranksatlevel{i}
\end{align*}
Then we have 
\begin{align*} 
\setweight{\prebsl{S}{f}} 
&\geq \sum_{i=2}^\maxrank 2^i\ranksatlevel{i} + \sum_{x\in\ranksatlevel 1} \nodeweight x\\ 
&\geq \sum_{i=2}^\maxrank 2^i\ranksatlevel{i} + 2\nodesatlevel{1}\\
&\geq\sum_{i=2}^\maxrank 2^i\ranksatlevel{i} + 2\left(2^{\intheight{\prebslsub{S}{f}}-1} - \sum_{i=2}^R 2^{i-1}\ranksatlevel{i}\right)\\ 
&\geq 2^{\intheight{\prebslsub{S}{f}}}\\ 
\log\setweight{\prebsl{S}{f}} &\geq \intheight{\prebsl{S}{f}}.  
\end{align*} 
\end{proof}

\begin{lemma} 
\label{lem:phasetwo} 
Phase~II of the \Union{} operation has a worst-case time complexity of \[\mathcal O\left(\log n + \sum_{i=2}^{\lasta-1} \intfun{A_i} +\sum_{j=2}^{\lastb-1} \intfun{B_j}\right).\] 
\end{lemma} 

\begin{proof} 
During Phase~II, \Bslfsplx{\intmax{A_i}} is invoked for $1\leq i< \lasta$, and \Bslfsplx{\intmax{B_j}} is invoked for $1\leq j< \lastb$. 
By Lemma~\ref{lem:wccostofbslfsplit}, 
where $L_{0} = \{A,B\}$,
and Lemma~\ref{lem:height}, the worst-case time complexity of Phase~II is 
\[ 
\mathcal O\left(\sum_{i=1}^{\lasta-1} (\log\intweight{A_i} - \min(\noderank{\intmax{A_i}},\noderank{\intmin{A_{i+1}}})+1)\right)+ 
\] 
\[\mathcal O\left(\sum_{j=1}^{\lastb-1} (\log\intweight{B_j} - \min(\noderank{\intmax{B_j}},\noderank{\intmin{B_{j+1}}})+1)\right)+
\mathcal O(\log n)\] 
Observe that $\nodenewweight{\intmin{B_i}}<\nodeweight{\intmin{A_{i+1}}}$ and $\nodenewweight{\intmin{A_{j+1}}}<\nodeweight{\intmin{B_{j+1}}}$. Then, by Lemma~\ref{lem:logncostofextended} and the definitions of \intfun{A_i} and \intfun{B_j}, the worst-case complexity of Phase~II is 
\[\mathcal O\left(\log n + \sum_{i=2}^{\lasta-1} \intfun{A_i} +\sum_{j=2}^{\lastb-1} \intfun{B_j}\right).\] 
\end{proof}

\begin{lemma} 
\label{lem:phasethree} 
Phase~III of the \Union{} operation has a worst-case time complexity of \[\mathcal O\left(\log n +\sum_{i=2}^{\lasta-1} \intfun{A_i}  +\sum_{j=2}^{\lastb-1} \intfun{B_j} \right).\] 
\end{lemma} 

\begin{proof} 
During Phase~III, we invoke 

\begin{itemize} 
\item \Bslfrewx{\intmin{A_i}}{\nodenewweight{\intmin{A_i}}} \quad for  $2\leq i \leq \lasta$, 
\item \Bslfrewx{\intmax{A_i}}{\nodenewweight{\intmax{A_i}}} \quad for  $1\leq i \leq \lastb$, 
\item \Bslfrewx{\intmin{B_j}}{\nodenewweight{\intmin{B_j}}} \quad for  $1\leq j \leq \lastb$, 
\item \Bslfrewx{\intmax{B_j}}{\nodenewweight{\intmax{B_j}}} \quad for  $1\leq j < \lasta$. 
\end{itemize}

Let $t_1 = \intmin{B}$. If \intnum{A}{B} is even, let $t_2=\intmax{A}$, otherwise let $t_2 = \intmax{B}$. Observe that $\nodenewrank{x}\leq\noderank x\leq\nodeheight{x}$ if and only if $x\not\in \{t_1,t_2\}$. 
Then, by Lemma~\ref{lem:costofbslfreweight} and Lemma~\ref{lem:height}, \Bslfrewx{x}{\nodenewweight{x}} has a worst-case time complexity of 
\begin{itemize} 
\item ${\cal O}(\log\intweight{A_i} - \nodenewrank{\intmin{A_i}}+1)$ for all 
$x\in \{\intmin{A_i}\,|\,1\leq i\leq\lasta\}$, 
\item ${\cal O}(\log\intweight{A_i} - \nodenewrank{\intmax{A_i}}+1)$ for all 
$x\in \{\intmax{A_i}\,|\,1\leq i\leq\lasta\}\setminus\{t_2\}$, 
\item ${\cal O}(\log\intweight{B_j} - \nodenewrank{\intmin{B_j}}+1)$ for all 
$x\in \{\intmin{B_j}\,|\,1\leq j\leq\lastb\}\setminus\{t_1\}$, 
\item ${\cal O}(\log\intweight{B_j} - \nodenewrank{\intmax{B_j}}+1)$ for all 
$x\in \{\intmax{B_j}\,|\,1\leq j\leq\lastb\}\setminus\{t_2\}$, 
\end{itemize} 
and a worst-case time complexity of $\mathcal O(\log n)$ for $x\in\{t_1,t_2\}$. 

Therefore, by Lemma~\ref{lem:logncostofextended} and the definitions of \intfun{A_i} and \intfun{B_j}, the worst-case time complexity of Phase~III is $\mathcal O\left(\log n + \sum_{i=2}^{\lasta-1} \intfun{A_i}  +\sum_{j=2}^{\lastb-1} \intfun{B_j} \right).$ 

\end{proof}

\begin{lemma} 
\label{lem:phasefour} 
Phase~IV of the \Union{} operation has a worst-case time complexity of \[\mathcal O\left(\log n+\sum_{i=2}^{\lasta-1} \intfun{A_i} +\sum_{j=2}^{\lastb-1} \intfun{B_j} \right).\] 
\end{lemma} 

\begin{proof} 

During Phase~IV, we invoke \Bslfjoinx{\intmax{A_{i}}}{\intmin{B_i}} for $1\leq i \leq \lastb$ and we invoke \Bslfjoinx{\intmax{B_{i}}}{\intmin{A_{i+1}}} for $1\leq i<z$. 
By Lemma~\ref{lem:wccostofbslfjoin}, where $L_{0}=\{A_{1},B_{1},A_{2},B_{2},\ldots\}$, and Lemma~\ref{lem:height}, the worst-case time complexity of Phase~IV is 
\[ 
\mathcal O\left(\sum_{i=2}^{\lasta} (\log\intweight{A_i} - \min(\noderank{\intmax{B_{i-1}}},\noderank{\intmin{A_i}},\noderank{\intmax{A_{i}}})+1)\right)+ 
\] 
\[\mathcal O\left(\sum_{j=1}^{\lastb} (\log\intweight{B_j} - \min(\noderank{\intmax{A_j}},\noderank{\intmin{B_j}},\noderank{\intmax{B_{j}}})+1)\right)
\] 
which by Lemma~\ref{lem:logncostofextended} and the definitions of \intfun{A_i} and \intfun{B_j} yields 
\[\mathcal O\left(\log n + \sum_{i=2}^{\lasta-1} \intfun{A_i} +\sum_{j=2}^{\lastb-1} \intfun{B_j}\right).\] 
\end{proof} 

The next theorem bounds the worst-case time complexity of the \Unionx{A}{B} operation.

\begin{thm} 
\label{thm:unionworstcase} 
The \Unionx{A}{B} operation has a worst-case time complexity of 
\[{\cal O}\left(\log n + \sum_{i=2}^{\lasta-1}\intfun{A_i} + \sum_{j=2}^{\lastb-1}\intfun{B_j} \right).\] 
\end{thm} 
\begin{proof} 
The worst-case time complexity of the \Unionx{A}{B} operation is determined by the time it spends on each of the four phases. 
Therefore, by Lemmas \ref{lem:phaseone}, \ref{lem:phasetwo}, \ref{lem:phasethree}, and \ref{lem:phasefour}, the theorem follows.

\end{proof}

\subsubsection{Amortized Time Complexity} 
\label{subsubsec:AnalysisUnionAmortized} 

Before we can show that the amortized time complexity of the \Unionx{A}{B} operation is ${\cal O}(\log n)$, we will need to prove 3 lemmas. Let us first define the potential loss associated with a gap. Recall the definitions of gaps $\agap{i}, \leftofa i, \rightofa i$ and similarly $ \bgap{j}, \leftofb j, \rightofb j$ first defined in Section~\ref{subsec:LogSQBound}. 

\begin{defn} 
We define \potlossa i and \potlossb j, the potential loss associated with gap \agap i and the potential loss associated with gap \bgap i respectively, for $1\leq i<z$ and $1\leq j<v$, as follows:

\begin{align} 
\label{eq:potloss} 
\potlossa i &= 2\log\agap i - \log \leftofa i - \log \rightofa i\\ 
\potlossb j &= 2\log\bgap j - \log \leftofb j - \log \rightofb j 
\end{align} 
Assume w.l.o.g.~that $\min(A) < \min(B)$. Then we also let $\potlossa 0 = 0$ and 
\begin{equation} 
\label{eq:potlossboundarybeginning} 
\potlossb 0 = -\log \leftofa 1. 
\end{equation} 
If $\max(A) >\max(B)$, then we have $\potlossa{\lasta} = 0$ and 
\begin{equation} 
\label{eq:potlossboundaryending1} 
\potlossb{\lastb} = -\log \rightofa{\lasta-1}. 
\end{equation} 
Otherwise, if $\max(A) <\max(B)$, we have $\potlossb{\lastb} = 0$ and 
\begin{equation} 
\label{eq:potlossboundaryending2} 
\potlossa{\lasta} = -\log \rightofb{\lastb-1}. 
\end{equation} 

\end{defn} 

Note that the potential loss associated with operation \Unionx{A}{B} is \cnsd{} times the sum of all defined \potlossa i and \potlossb j, where \cnsd{} is the constant in the potential function.

\begin{lemma} 
\label{lem:gapratio} 
Consider gap \agap i for any $1\leq i< \lasta$. Let $\maxgap{\agap i} = \max(\leftofa i, \rightofa i)$ and $\mingap{\agap i} = \min(\leftofa i, \rightofa i)$. Then we have 
\[ 
2^{-\potlossa i} \leq\frac{\agap i}{\maxgap{\agap i}} \leq \frac{\agap i}{\mingap{\agap i}} \leq 2^{\potlossa i}\quad \text{and} \quad 2^{-\potlossa i} \leq \frac{\maxgap{\agap i}}{\mingap{\agap i}}\leq 2^{\potlossa i}. 
\] 
\noindent Similarly, consider gap \bgap j for any $1\leq j< \lastb$, where $\maxgap{\bgap j} = \max(\leftofb j, \rightofb j)$ and $\mingap{\bgap j} = \min(\leftofb j, \rightofb j)$. Then we have 
\[ 
2^{-\potlossb j} \leq\frac{\bgap j}{\maxgap{\bgap j}} \leq \frac{\bgap j}{\mingap{\bgap j}} \leq 2^{\potlossb j}\quad \text{and} \quad 2^{-\potlossb j} \leq \frac{\maxgap{\bgap j}}{\mingap{\bgap j}}\leq 2^{\potlossb j}. 
\] 
\end{lemma}

\begin{proof} 
We will prove the first part dealing with gap \agap i. The proof of the second part is symmetric. Assume w.l.o.g.~that $\leftofa i \leq \rightofa i$. By definition, we have
\begin{equation*} 
\label{eq:gapratioupperbound} 
\frac{\agap i}{\maxgap{\agap i}} \leq \frac{\agap i}{\mingap{\agap i}}.
\end{equation*} 
Since we have $2\log\agap i -\log \maxgap{\agap i}- \log \mingap{\agap i} = \potlossa i$ this yields
\begin{align*} 
\log\frac{\agap i}{\mingap{\agap i}} +\log\frac{\agap i}{\maxgap{\agap i}} &= \potlossa i\\
\log\frac{\agap i}{\mingap{\agap i}} &\leq \potlossa i\\
\frac{\agap i}{\mingap{\agap i}} &\leq 2^{\potlossa i}.
\end{align*} 
On the other hand, since $\mingap{\agap i}\leq\maxgap{\agap i}\leq \agap i$, we have
\begin{align*} 
\frac{\agap i}{\maxgap{\agap i}}&\leq 2^{\potlossa i}\\ 
\frac{1}{2^{\potlossa i}} &\leq \frac{\maxgap{\agap i}}{\agap i}\\ 
\frac{1}{2^{\potlossa i}} &\leq \frac{\maxgap{\agap i}}{\agap i} \leq \frac{\agap i}{\maxgap{\agap i}}
\end{align*} 
Lastly, observe that 
\begin{align*} 
\frac{\maxgap{\agap i}}{\mingap{\agap i}} \leq \frac{\agap i}{\mingap{\agap i}} \leq 2^{\potlossa i}
\end{align*} 
and 
\begin{align*} 
\frac{1}{2^{\potlossa i}} \leq \frac{\maxgap{\agap i}}{\agap i} &\leq \frac{\maxgap{\agap i}}{\mingap{\agap i}}
\end{align*} 
This concludes the proof. Second part of the lemma can be proven using symmetric arguments. 

\end{proof}

\begin{lemma} 
\label{lem:costeqpotloss} 
Let $\intsubs{A}{B} = \{\intl A_1, \intl B_1,\intl A_2, \intl B_2, \ldots\}$ be the set of \segments{} of  with respect to operation \Unionx{A}{B}. For any $i$ and $j$, where $1< i <\lasta$ and $1< j < \lastb$, let 
\[\maxpotlossa i = \max\left(2^{\potlossa{i-2}}, 2^{\potlossb{i-2}},2^{\potlossa{i-1}}, 2^{\potlossb{i-1}}, 2^{\potlossa i},2^{\potlossb i}, 2^{\potlossa{i+1}}\right)\] 
and 
\[\maxpotlossb j = \max\left(2^{\potlossb{j-2}},2^{\potlossa{j-1}}, 2^{\potlossb{j-1}}, 2^{\potlossa j},2^{\potlossb j}, 2^{\potlossa{j+1}},2^{\potlossb{j+1}}\right).\] 
Then for $1< i <\lasta$ and $1< j < \lastb$, we have 
\[\intfun{A_i} = \mathcal O(\log \maxpotlossa i)\quad \mbox{ and } \quad \intfun{B_j} = \mathcal O(\log \maxpotlossb j).\] 
\end{lemma}

\begin{proof}

We will present the proof of the first equality. The proof of the second one is analogous. Let $\rightof{a_{i-1}} = \leftof{b_{i-1}} = x$. Then, by Lemma~\ref{lem:gapratio} we have 

\[ 
\frac{1}{2^{\potlossa{i-1}}} \leq \frac{\leftof{a_{i-1}}}{\rightof{a_{i-1}}} \leq 2^{\potlossa{i-1}} 
\quad \mbox{ and }\quad 
\frac{1}{2^{\potlossb{i-1}}} \leq \frac{\rightof{b_{i-1}}}{\leftof{b_{i-1}}} \leq 2^{\potlossb{i-1}} 
\] 
which imply 
\begin{equation} 
\label{eq:robim0}
\frac{x}{\maxpotlossa i}\leq\frac{x}{2^{\potlossa{i-1}}} \leq \leftof{a_{i-1}} \leq x2^{\potlossa{i-1}} \leq x \maxpotlossa i 
\end{equation} 
and 
\begin{equation} 
\label{eq:robim1} 
\frac{x}{\maxpotlossa i}\leq\frac{x}{2^{\potlossb{i-1}}} \leq \rightof{b_{i-1}} \leq x2^{\potlossb{i-1}} \leq x \maxpotlossa i 
\end{equation} 
Note that \leftof{a_i} = \rightof{b_{i-1}}, and \rightof{b_{i-2}} = \leftof{a_{i-1}}. By Lemma~\ref{lem:gapratio}, we have 
\begin{equation*}
\frac{1}{2^{\potlossa{i}}} \leq \frac{\rightof{a_{i}}}{\leftof{a_{i}}} \leq 2^{\potlossa{i}} 
\quad \mbox{ and }\quad 
\frac{1}{2^{\potlossb{i-2}}} \leq \frac{\leftof{b_{i-2}}}{\rightof{b_{i-2}}} \leq 2^{\potlossb{i-2}} 
\end{equation*}
which by (\ref{eq:robim0}) and (\ref{eq:robim1}) imply 
\begin{equation} 
\label{eq:robim2}
\frac{x}{\maxpotlossa i^2}\leq\frac{1}{2^{\potlossa{i}}}\cdot\frac{x}{\maxpotlossa i}\leq \rightof{a_{i}} \leq 2^{\potlossa{i}} \cdot x\maxpotlossa i  \leq x \maxpotlossa i^2 
\end{equation} 
and 
\begin{equation}
\label{eq:robim3} 
\frac{x}{\maxpotlossa i^2}\leq\frac{1}{2^{\potlossb{i-2}}}\cdot\frac{x}{\maxpotlossa i}\leq \leftof{b_{i-2}} \leq 2^{\potlossb{i-2}} \cdot x\maxpotlossa i  \leq x \maxpotlossa i^2 
\end{equation} 
Note that \rightof{a_{i-2}} = \leftof{b_{i-2}}, and \leftof{b_{i}} = \rightof{a_{i}}. By Lemma~\ref{lem:gapratio}, we have 
\[ 
\frac{1}{2^{\potlossa{i-2}}} \leq \frac{\leftof{a_{i-2}}}{\rightof{a_{i-2}}} \leq 2^{\potlossa{i-2}} 
\quad \mbox{ and }\quad 
\frac{1}{2^{\potlossb{i}}} \leq \frac{\rightof{b_{i}}}{\leftof{b_{i}}} \leq 2^{\potlossb{i}} 
\] 
which by (\ref{eq:robim2}) and (\ref{eq:robim3}) imply 
\begin{equation} 
\label{eq:robim4}
\frac{x}{\maxpotlossa i^3}\leq\frac{1}{2^{\potlossa{i-2}}}\cdot\frac{x}{\maxpotlossa i^2} \leq \leftof{a_{i-2}} \leq 2^{\potlossa{i-2}}\cdot x\maxpotlossa i^2 \leq x \maxpotlossa i^3 
\end{equation} 
and 
\begin{equation}
\label{eq:robim5} 
\frac{x}{\maxpotlossa i^3}\leq\frac{1}{2^{\potlossb{i}}}\cdot\frac{x}{\maxpotlossa i^2} \leq \rightof{b_{i}} \leq 2^{\potlossb{i}}\cdot x\maxpotlossa i^2 \leq x \maxpotlossa i^3 
\end{equation} 
Note that \leftof{a_{i+1}} = \rightof{b_{i}}. By Lemma~\ref{lem:gapratio}, we have 
\[ 
\frac{1}{2^{\potlossa{i+1}}} \leq \frac{\rightof{a_{i+1}}}{\leftof{a_{i+1}}} \leq 2^{\potlossa{i+1}} 
\] 
which by (\ref{eq:robim4}) and (\ref{eq:robim5}) implies 
\begin{equation} 
\label{eq:robim6}
\frac{x}{\maxpotlossa i^4}\leq\frac{1}{2^{\potlossa{i+1}}}\cdot\frac{x}{\maxpotlossa i^3}\leq \rightof{a_{i+1}} \leq 2^{\potlossa{i+1}} \cdot x\maxpotlossa i^3  \leq x \maxpotlossa i^4 
\end{equation} 
Similarly, by Lemma~\ref{lem:gapratio}, we have 
\[ 
\frac{1}{2^{\potlossa{i-1}}} \leq \frac{a_{i-1}}{\rightof{a_{i-1}}} \leq 2^{\potlossa{i-1}} 
\quad \mbox{ and }\quad 
\frac{1}{2^{\potlossb{i-1}}} \leq \frac{b_{i-1}}{\leftof{b_{i-1}}} \leq 2^{\potlossb{i-1}} 
\] 
which imply 
\begin{equation} 
\label{eq:aim1} 
\frac{x}{\maxpotlossa i}\leq\frac{x}{2^{\potlossa{i-1}}} \leq a_{i-1} \leq x2^{\potlossa{i-1}} \leq x \maxpotlossa i 
\end{equation} 
and 
\begin{equation} 
\label{eq:bim1} 
\frac{x}{\maxpotlossa i}\leq\frac{x}{2^{\potlossb{i-1}}} \leq b_{i-1} \leq x2^{\potlossb{i-1}} \leq x \maxpotlossa i 
\end{equation} 
By Lemma~\ref{lem:gapratio}, we have 
\[ 
\frac{1}{2^{\potlossa{i}}} \leq \frac{a_{i}}{\leftof{a_{i}}} \leq 2^{\potlossa{i}} 
\quad \mbox{ and }\quad 
\frac{1}{2^{\potlossb{i-2}}} \leq \frac{b_{i-2}}{\rightof{b_{i-2}}} \leq 2^{\potlossb{i-2}} 
\] 
which by (\ref{eq:aim1}) and (\ref{eq:bim1}) imply 
\begin{equation} 
\label{eq:ai} 
\frac{x}{\maxpotlossa i^2}\leq\frac{1}{2^{\potlossa{i}}}\cdot\frac{x}{\maxpotlossa i}\leq a_{i} \leq 2^{\potlossa{i}} \cdot x\maxpotlossa i  \leq x \maxpotlossa i^2 
\end{equation} 
and 
\begin{equation} 
\label{eq:bi} 
\frac{x}{\maxpotlossa i^2}\leq\frac{1}{2^{\potlossb{i-2}}}\cdot\frac{x}{\maxpotlossa i}\leq b_{i-2} \leq 2^{\potlossb{i-2}} \cdot x\maxpotlossa i  \leq x \maxpotlossa i^2 
\end{equation} 
By Lemma~\ref{lem:gapratio}, we have 
\[ 
\frac{1}{2^{\potlossa{i-2}}} \leq \frac{a_{i-2}}{\rightof{a_{i-2}}} \leq 2^{\potlossa{i-2}} 
\quad \mbox{ and }\quad 
\frac{1}{2^{\potlossb{i}}} \leq \frac{b_{i}}{\leftof{b_{i}}} \leq 2^{\potlossb{i}} 
\] 
which by (\ref{eq:ai}) and (\ref{eq:bi}) imply 
\begin{equation} 
\label{eq:aim2} 
\frac{x}{\maxpotlossa i^3}\leq\frac{1}{2^{\potlossa{i-2}}}\cdot\frac{x}{\maxpotlossa i^2} \leq a_{i-2} \leq 2^{\potlossa{i-2}}\cdot x\maxpotlossa i^2 \leq x \maxpotlossa i^3 
\end{equation} 
and 
\begin{equation} 
\label{eq:bim2}
\frac{x}{\maxpotlossa i^3}\leq\frac{1}{2^{\potlossb{i}}}\cdot\frac{x}{\maxpotlossa i^2} \leq b_{i} \leq 2^{\potlossb{i}}\cdot x\maxpotlossa i^2 \leq x \maxpotlossa i^3 
\end{equation} 
By Lemma~\ref{lem:gapratio}, we have 
\[ 
\frac{1}{2^{\potlossa{i+1}}} \leq \frac{a_{i+1}}{\leftof{a_{i+1}}} \leq 2^{\potlossa{i+1}} 
\] 
which by (\ref{eq:aim2}) implies 
\begin{equation} 
\label{eq:aip1} 
\frac{x}{\maxpotlossa i^4}\leq\frac{1}{2^{\potlossa{i+1}}}\cdot\frac{x}{\maxpotlossa i^3}\leq a_{i+1} \leq 2^{\potlossa{i+1}} \cdot x\maxpotlossa i^3  \leq x \maxpotlossa i^4. 
\end{equation}

We proceed as follows. For $1<i<z$, we have 
{\allowdisplaybreaks 
\begin{align*} 
\intfun{A_i} 
&= 
\log
\frac
{\nodeweight{\intmax{A_{i-1}}}+\nodeweight{\intmin{A_{i+1}}} + \sum_{\node x\in \intl A_i}{\nodeweight x}}%
{\min(\nodenewweight{\intmax{{B_{i-1}}}}, \nodenewweight{\intmin{{A_i}}}, \nodenewweight{\intmax{{A_i}}}, \nodenewweight{\intmin{{B_{i}}}})}
\\ 
&\leq 
\log
\frac
{\gap a_{i-2} +\gap a_{i-1} + \gap b_{i-2}+\nodeweight{\intmin{A_{i+1}}} + \sum_{\node x\in \intl A_i}{\nodeweight x}}%
{\min(\nodenewweight{\intmax{{B_{i-1}}}}, \nodenewweight{\intmin{{A_i}}}, \nodenewweight{\intmax{{A_i}}}, \nodenewweight{\intmin{{B_{i}}}})}
\\
&\leq 
\log
\frac
{\gap a_{i-2} +\gap a_{i-1} +\gap a_i + \gap a_{i+1} + \gap b_{i-2}+ \gap b_{i}+ \sum_{\node x\in \intl A_i}{\nodeweight x}}%
{\min(\nodenewweight{\intmax{{B_{i-1}}}}, \nodenewweight{\intmin{{A_i}}}, \nodenewweight{\intmax{{A_i}}}, \nodenewweight{\intmin{{B_{i}}}})}
\\
&\leq 
\log
\frac{\gap a_{i-2} +\gap a_{i-1} +\gap a_i + \gap a_{i+1} + \gap b_{i-2}+ \gap b_{i} + \gap 2b_{i-1}}
{\min(\nodenewweight{\intmax{{B_{i-1}}}}, \nodenewweight{\intmin{{A_i}}}, \nodenewweight{\intmax{{A_i}}}, \nodenewweight{\intmin{{B_{i}}}})}
\\
&\leq 
\log
\frac{\gap a_{i-2} +\gap a_{i-1} +\gap a_i + \gap a_{i+1} + \gap b_{i-2}+ \gap b_{i}+ \gap 2b_{i-1}}
{\min(\leftof{b_{i-1}}, \rightof{a_{i-1}}, \leftof{a_i},\rightof{b_{i-1}})}
\\ 
&= 
\mathcal O\left(
\log\frac{x\maxpotlossa i^4}{\min(\leftof{b_{i-1}}, \rightof{a_{i-1}}, \leftof{a_i},\rightof{b_{i-1}})}
\right)
&&\mbox{by (\ref{eq:aim1})$-$(\ref{eq:aip1}).}
\\ 
&= \mathcal O\left(\log\frac{x\maxpotlossa i^4}{x/\maxpotlossa i}\right)&&\mbox{by (\ref{eq:robim1})}.\\ 
&= \mathcal O(\log \maxpotlossa i). 
\end{align*} 
} 
The proof of the second equality, $\intfun{B_j} = \mathcal O(\log \maxpotlossb j)$ for $1<j<\lastb$, is analogous. 

\end{proof}

\begin{lemma} 
\label{lem:maxmappingbound} 
For $1< i < \lasta$ and $1< j < \lastb$, we have 
\[ 
7\cdot\left(\sum_i \log 2^{\potlossa i}+ \sum_j \log 2^{\potlossb j}\right) > \left(\sum_i \log\maxpotlossa i + \sum_j \log\maxpotlossb j\right). 
\] 
\end{lemma} 

\begin{proof} 
Observe that a gap $\agap k$ can be mapped to at most seven times by unique $\maxpotlossa i$'s and $\maxpotlossb j$'s; namely only by $\maxpotlossa{k-1}, \maxpotlossb{k-1}, \maxpotlossa{k}, \maxpotlossb{k}, \maxpotlossa{k+1}, \maxpotlossb{k+1}, \maxpotlossa{k+2}$. Similarly, a gap \bgap{k} can be mapped to at most seven times by unique $\maxpotlossa i$'s and $\maxpotlossb j$'s; namely only by $\maxpotlossb{k-1}, \maxpotlossa{k}, \maxpotlossb{k}, \maxpotlossa{k+1}, \maxpotlossb{k+1}, \maxpotlossa{k+2}, \maxpotlossb{k+2}$. The lemma follows. 
\end{proof}

We are now ready to bound the amortized time complexity of the \Union{} operation.

\begin{thm} 
\label{thm:unionamortizedcase} 
The \Unionx{A}{B} operation has an amortized time complexity of ${\cal O}(\log n)$. 
\end{thm} 
\begin{proof} 
We will analyze the \Union{} operation using the potential method \cite{amortizedcomplexity}. Recall that \datast i represent the data structure after operation $i$, where $D_0$ is the initial data structure. The amortized cost of operation $i$ is $\amcost i = \actcost i + \potfun{\datast i} - \potfun{\datast{i-1}}$. Then the amortized cost of the \Unionx{A}{B} operations is

{\allowdisplaybreaks
\begin{align*} 
\amcost i &= \actcost i + \Delta\Phi \\ 
&\leq \cnse\left(\log n +\sum_{i=2}^{\lasta-1} \intfun{A_i} + \sum_{j=2}^{\lastb-1} \intfun{B_j}\right)+ \Delta\Phi&&\mbox{by Theorem~\ref{thm:unionworstcase}.}\\ 
&=\mathcal O(\log n) + \cnse\left(\sum_{i=2}^{\lasta-1} \intfun{A_i} + \sum_{j=2}^{\lastb-1} \intfun{B_j}\right)+ \Delta\Phi\\ 
&=\mathcal O(\log n) + \cnse\left(\sum_{i=2}^{\lasta-1} \intfun{A_i} + \sum_{j=2}^{\lastb-1} \intfun{B_j}\right)- \cnsd\cdot\left(\sum_{i=0}^{\lasta} \potlossa{i} + \sum_{j=0}^{\lastb} \potlossb{j}\right)&&\mbox{by (\ref{eq:mainpotential})}\\
&=\mathcal O(\log n) + \cnse\left(\sum_{i=2}^{\lasta-1} \intfun{A_i} + \sum_{j=2}^{\lastb-1} \intfun{B_j}\right)- \cnsd\cdot\left(\sum_{i=2}^{\lasta-1} \potlossa{i} + \sum_{j=2}^{\lastb-1} \potlossb{j}\right)&& \mbox{$\mathsf{pl}(\cdot)={\cal O}(\log n).$ }\\ 
&<\mathcal O(\log n) + \cnse\left(\sum_{i=2}^{\lasta-1} \intfun{A_i} + \sum_{j=2}^{\lastb-1} \intfun{B_j}\right)- \frac{\cnsd}{7}\cdot\left(\sum_{i=2}^{\lasta-1} \log\maxpotlossa i + \sum_{j=2}^{\lastb-1} \log\maxpotlossb j\right)&& \mbox{by Lemma~\ref{lem:maxmappingbound}.}\\ 
&\leq\mathcal O(\log n) + \cnsf\left(\sum_{i=2}^{\lasta-1} \log\maxpotlossa i + \sum_{j=2}^{\lastb-1} \log\maxpotlossb j\right)- \frac{\cnsd}{7}\cdot\left(\sum_{i=2}^{\lasta-1} \log\maxpotlossa i + \sum_{j=2}^{\lastb-1} \log\maxpotlossb j\right)&& \mbox{by Lemma~\ref{lem:costeqpotloss}.}\\ 
&\leq\mathcal O(\log n) && \mbox{Set ($\cnsd = 7\cnsf$).} 
\end{align*}}
\end{proof} 
We can now state our main theorem. 

\begin{thm} 
\label{thm:maintheorem} 
The \Ds{} executes a sequence of $m$ \hide{\Ms{}, }\Find{}, \Src{}, \Spl{}, and \Union{} operations \hide{, $n$ of which are \Ms{} operations, }in worst-case ${\cal O}(m\log n)$ time. 
\end{thm} 

\begin{proof} 
Follows directly from Theorem~\ref{thm:otheropsbound} and Theorem~\ref{thm:unionamortizedcase}. 
\end{proof}

\bibliographystyle{plain}
\bibliography{lgfussbib}

\end{document}